\documentclass[letterpaper, 10pt, twocolumn, journal]{IEEEtran}

\IEEEoverridecommandlockouts  
\usepackage[utf8]{inputenc}
\usepackage[version=4]{mhchem}
\usepackage{colonequals}
\newcommand*{\logeq}{\Leftrightarrow}
\usepackage{pgfplots}
\pgfplotsset{compat=1.16}
\usepackage{tikz}
\usepackage{dirtytalk}
\usepackage{centernot}
\usepackage{xcolor}
\usepackage{graphicx}
\usepackage{graphics}
\usepackage{subcaption,comment}
\usepackage{amsmath}
\usepackage[version=4]{mhchem}
\usepackage{siunitx}
\usepackage{longtable,tabularx}
\usepackage{algorithm}
\usepackage{algorithmicx}
\usepackage{algpseudocode}
\usepackage{amsmath}
\usepackage{amssymb}

\usepackage{proof}
\usepackage{epsfig}
\usepackage{calc}

\usepackage{amsthm}
\usepackage{accents}
\usepackage{cite}
\usepackage{hyperref}

\newtheorem{remark}{Remark}
\newtheorem{lemma}{Lemma}
\newtheorem{definition}{Definition}
\newtheorem{corollary}{Corollary}
\newtheorem{assumption}{Assumption}
\newtheorem{theorem}{Theorem}
\newtheorem{example}{Example}
\newcommand{\bx}{{\boldsymbol{x}}}
\newcommand{\iin}{{\mathcal{I}_N}}
\newcommand{\pn}{{\mathcal{P}_{\boldsymbol{x}}}}
\newcommand{\qn}{{\mathcal{Q}_{\boldsymbol{x}}}}
\newcommand{\pu}{{\mathcal{P}_{\boldsymbol{u}}}}
\newcommand{\qu}{{\mathcal{Q}_{\boldsymbol{u}}}}

\newcommand{\ix}{{\mathcal{I}_{\boldsymbol{x}}}}
\newcommand{\iu}{{\mathcal{I}_{\boldsymbol{u}}}}

\newcommand{\bu}{{\boldsymbol{u}}}
\newcommand{\by}{{\boldsymbol{y}}}
\newcommand{\bz}{{\boldsymbol{z}}}

\newcommand{\bv}{{\boldsymbol{v}}}

\title{\LARGE \bf
Universal Barrier Functions for Safety and Stability of Constrained Nonlinear Systems
}

\author{Vrushabh Zinage, \and Efstathios Bakolas
\thanks{Vrushabh Zinage (PhD Candidate) and Efstathios Bakolas (Associate Professor) are with the Department of Aerospace Engineering and Engineering Mechanics,
The University of Texas at Austin, Texas 78712-1221, USA, 
{\tt\small \{vrushabh.zinage,bakolas\}@austin.utexas.edu}}
}

\begin{document}

\bibliographystyle{IEEEtran} 

\maketitle

\begin{abstract}
In this paper, we address the problem of synthesizing safe and stabilizing controllers for nonlinear systems subject to complex safety specifications and input constraints. We introduce the Universal Barrier Function (UBF), a single continuously differentiable scalar-valued function that encodes both stability and safety criteria while accounting for input constraints. Using the UBF, we formulate a Quadratic Program (UBF-QP) to generate control inputs that are both safe and stabilizing under input constraints. We demonstrate that the UBF-QP is feasible if a UBF exists. Furthermore, under mild conditions, we prove that a UBF always exists. The proposed framework is then extended to systems with higher relative degrees. Finally, numerical simulations illustrate the effectiveness of our proposed approach. The code is available at \href{https://github.com/Vrushabh27/ubf}{\textbf{https://github.com/Vrushabh27/ubf}}
\end{abstract}

\section{Introduction\label{sec:intro}}
Model Predictive Control (MPC) has been widely adopted for control design in various real-world applications, including quadrotors \cite{wang2021efficient_mpc_quadrotor_1,bangura2014real_mpc_quadrotor_2}, legged robots \cite{rathod2021model_mpc_legged_1,farshidian2017real_mpc_legged_2,zinage2024transformermpc_legged_3}, humanoid robots \cite{katayama2023model_mpc_humanoids_1,best2016new_mpc_humanoid_2}, multi-agent systems  and manipulators \cite{kleff2021high_mpc_manipulator_1}. However, in recent years, Control Barrier Function-based Quadratic Programs (CBF-QP) have emerged as a promising alternative to MPC based controllers, that offers a computationally efficient synthesis of safe control inputs. CBF-QP is computationally efficient compared to MPC because it formulates control synthesis as a simpler quadratic program that directly enforces safety constraints without requiring iterative optimization over a prediction horizon, thereby reducing computational complexity. 

CBF-QP\cite{ames2014control_bf_6} and Control Lyapunov Function-Control Barrier Function based Quadratic Programs (CLF-CBF-QP) \cite{ames2014control_bf_6} approaches have gained popularity for generating safe and stabilizing controllers for control-affine nonlinear systems. Their applicability covers systems with input constraints \cite{agrawal2021safe_input_constrained_1,zinage2023neural_input_constrained_2}, higher relative degrees\cite{xiao2021high_order_cbf_1,tan2021high_order_2,zinage2024disturbance_high_order}, hybrid systems\cite{marley2024hybrid_systems_1,lindemann2021learning_hybrid_2}, unknown nonlinear systems \cite{zinage2023neural_unknown,zinage2023neural_unknown_3,jagtap2020control_unknown_2}, sampled data systems\cite{breeden2021control_sampled_data_systems_1,niu2021safety_sampled_data_systems_2,taylor2022safety_sampled_data_systems_3,oruganti2023robust_sampled_data_systems_4,zinage2024decentralized_multi_agent_systems_5}, input-delay systems\cite{jankovic2018control_input_delay_1}, and multi-agent systems\cite{jankovic2023multiagent_systems_1,lindemann2019control_multi_agent_systems_2,wang2017safety_multi_agents_3,glotfelter2017non_smooth_muliti_agent_systems_4}. These methods have also found practical applications in robot manipulation\cite{cortez2019control_robotic_manipulation_1}, bipedal robotics\cite{hsu2015control_bipedal_1}, and verification and control \cite{clark2021verification}.

Most CLF-CBF-QP-based methods ensure the feasibility of the Quadratic Program (QP) by introducing slack variables for the Control Lyapunov Function (CLF) condition or both CLF and CBF conditions. Slack variables are additional terms that relax certain constraints, allowing for a feasible solution even when strict adherence to the original conditions is not possible. This helps ensure that a feasible control input can be found, although it may come at the cost of reduced stability or safety guarantees. To address this limitation, \cite{romdlony2016stabilization_clbf} integrates a CLF and a CBF to construct Control Lyapunov-Barrier Function (CLBF). This method then employs Sontag's universal formula to generate a smooth controller. However, \cite{braun2017existence,braun2020existence_comment} shows that the CLBF does not exist under the assumptions presented in  \cite{romdlony2016stabilization_clbf} and does not guarantee stability or convergence to the desired equilibrium point (the origin). On the other hand, if constraints related to CLF and CBF conditions are strictly enforced for stability and safety, the resulting inputs synthesized by solving the resulting CLF-CBF-QP can be non-Lipschitz (as proven via Robinson's counterexample \cite{robinson1982generalized_robinson}) or even infeasible. This consequently questions the well posedness of the closed-loop system. Moreover, in the presence of input constraints, even with known CLF and CBF, pointwise input feasibility for a given state cannot be guaranteed, as there may not exist a control input that simultaneously satisfies both input constraints and CLF (or CBF) conditions. This challenge is further compounded when dealing with complex safety specifications and systems with input constraints and higher relative degrees with respect to the CBF.

Several studies have addressed multiple safety constraints by employing various approaches. These include imposing multiple CBF constraints on inputs\cite{rauscher2016constrained_multiple_cbf_1}, switching between CBFs with non-intersecting super-level sets\cite{xu2018constrained_multiple_cbf_2}, and ensuring feasible inputs using multiple CBFs\cite{cortez2021robust_multiple_cbf_3} or Lyapunov barrier functions\cite{breeden2023compositions_multiple_cbf_4,liu2019barrier_multiple_cbf_5}. However, a key limitation of these works is their focus on the intersection of safety sets, which can be overly conservative, especially when dealing with complex safety and input constraints (such as when overall constraints are represented by unions and intersections of simpler safe sets). The primary motivation for this work is to develop a method that can simultaneously ensure stability, safety, and input constraint satisfaction, while also designing inputs for systems with higher relative degrees and general nonlinear controlled systems (including control-affine systems as a special case). Unlike existing approaches, the proposed method provides practical implementation benefits by reducing conservatism in control design and offering improved feasibility under input constraints. Additionally, it provides theoretical guarantees of stability and safety by utilizing a single scalar-valued, continuously differentiable function.
The contributions of the paper are as follows.
\begin{enumerate}
    \item We propose the notion of scalar-valued continuously differentiable Universal Barrier Functions (UBF) that unify the notion of stability and safety for input constrained nonlinear systems (Definition \ref{eqn:ubf}).
    \item Next, we formulate a UBF-based quadratic program (UBF-QP) to synthesize safe and stabilizing control inputs under given complex state and input constraints specifications. Furthermore, under these specifications, we provide sufficient conditions for the feasibility of UBF-QP (Theorem \ref{thm:ubf_qp_feasible}).
    \item Under some mild conditions, we provide sufficient conditions for the existence of a UBF.
    \item We propose High Order UBF that extends the notion of UBF for systems with higher relative degrees. Furthermore, we show that High-Order UBF is more general than the notions of High Order CBF's proposed in the literature \cite{xiao2021high_order_cbf_1,tan2021high_order_2}.
\end{enumerate}

The overall structure of the paper is as follows. In Section \ref{sec:nomenclature}, we present the nomenclature that we use in this paper followed by preliminaries in Section \ref{sec:prelim}. Section \ref{sec:motivation} presents the motivation for this work followed by the main contributions in Section \ref{sec:main_results}. Finally, Section \ref{sec:numerical_simulations} presents numerical simulations followed by concluding remarks in Section \ref{sec:conclusion}.
\section{Nomenclature\label{sec:nomenclature}}
Vectors and matrices are denoted by bold and capital letters, respectively. The sets of real numbers, non-negative real numbers, and positive integers are denoted by $\mathbb{R}$, $\mathbb{R}^+$, and $\mathbb{Z}_{>0}$, respectively. We denote by $\mathcal{I}_{(a,b)} := \{a, a+1, \dots, b\}$ where $a$ and $b (\geq a)$ are integers. For brevity, we denote by $\mathcal{I}_{N} := \mathcal{I}_{(0,N)}$ for some non-negative integer $N$. The $n$-dimensional zero vector is denoted by $\mathbf{0}_n$. Denote two disjoint sets by $\pn$ and $\qn$ such that $\pn \cup \qn = \iin$.
Furthermore, given $\pn$ and $\qn$, we define the operator $\oplus_i$ for $i\in \iin$ as follows
\begin{align}
     \oplus_i = \begin{cases} \cup  & \text { if } i\in \pn \\  \cap  & \text { if } i\in \qn
     \end{cases}\nonumber
    \end{align}
For any discrete and finite set $\mathcal{A}$, $|\mathcal{A}|$ denotes its cardinality. In addition, $P(\mathcal{A})$ denotes the power set of $\mathcal{A}$. The notation $\|\cdot\|$ denotes the Euclidean norm of a vector. $\operatorname{Int}(\mathcal{S})$ and $\partial \mathcal{S}$ denote the interior and boundary of the set $\mathcal{S}$, respectively. $L_f V$ denotes the Lie derivative of a scalar-valued differentiable function $V$ with respect to the function $f$.
For a scalar $a \in \mathbb{R}$, a function $\alpha: [0, a) \rightarrow [0, \infty)$ is of class $\mathcal{K}$ if it is continuous, strictly increasing, and $\alpha(0) = 0$. A function $\alpha: [0, \infty) \rightarrow [0, \infty)$ is of class $\mathcal{K}_\infty$ if it is of class $\mathcal{K}$ and $\underset{{r \rightarrow \infty}}{\lim}\alpha(r) = \infty$. A continuous function $\beta: [0, a) \times [0, \infty) \rightarrow [0, \infty)$ is of class $\mathcal{KL}$ if, for each fixed $s$, $\beta(\cdot, s)$ is of class $\mathcal{K}$, and for each fixed $r$, $\beta(r, \cdot)$ is decreasing with $\lim_{s \rightarrow \infty} \beta(r, s) = 0$. A scalar-valued function $V : \mathbb{R}^n \rightarrow \mathbb{R}^+$ is said to be positive definite if there exists a class $\mathcal{K}$ function $\alpha$ such that $V(\bx) \ge \alpha(|\bx|)$ for $\bx \in \mathbb{R}^n$.
For a scalar-valued function $c : \mathbb{R} \rightarrow \mathbb{R}$, $c^{-1}(r)$ denotes its inverse, i.e., $c(c^{-1}(r)) = 1$ for some $r \in \mathbb{R}$. In addition, $(c(r))^{-1} = \frac{1}{c(r)}$.


\section{Preliminaries\label{sec:prelim}}
Consider the nonlinear system given by
\begin{align}    
\dot{\bx}=F(\bx,\bu),\quad \bx(0)=\bx_0
\label{eqn:nonlinear_system_dynamics}
\end{align}
where $\bx\in\mathbb{R}^n$, $\bu\in\mathcal{U}\subset\mathbb{R}^m$ ($\mathcal{U}$ is a compact set) and $F:\mathbb{R}^n\times\mathbb{R}^m\rightarrow\mathbb{R}^n$
is continuously differentiable. In addition, we denote by $\bu^j$ $(j\in \mathcal{I}_{(1,m)})$, the $j^{\text{th}}$ component of $\bu$. In addition, for defining the notion of CLF and CBF, we consider the following control-affine nonlinear system as
\begin{align}    
\dot{\bx}=f(\bx)+g(\bx)\bu,\quad \bx(0)=\bx_0
\label{eqn:nonlinear_system_dynamics_affine}
\end{align}
where $f:\mathbb{R}^n\rightarrow\mathbb{R}^n$ and $g:\mathbb{R}^n\rightarrow\mathbb{R}^{n\times m}$ are Lipschitz continuous functions. We assume the origin to be the unique equilibrium point for the (unforced) systems \eqref{eqn:nonlinear_system_dynamics} and \eqref{eqn:nonlinear_system_dynamics_affine}. For brevity, we drop the time indexing of $\bx$ and $\bu$ wherever it is not required. In the following, we briefly discuss the notions of CLF and CBF.
\begin{definition}
    \normalfont \textbf{(CLF)} Given an open set $\mathcal{D}\subset\mathbb{R}^n$ and $\mathbf{0}_n\in\mathcal{D}$,  a scalar valued continuously differentiable function $V:\mathbb{R}^n\rightarrow\mathbb{R}^+$ is said to be a Control Lyapunov Function (CLF) if $V$ satisfies the following conditions:
    \begin{itemize}
        \item $V$ is a positive definite function
        \item $V$ is proper in $\mathcal{D}$, that is, the set $\{\bx\in\mathcal{D}:\;V(\bx)\leq c\}$ is a compact set for all $c>0$
        \item there exist a positive definite function $P(\bx):\mathbb{R}^n\rightarrow\mathbb{R}^+$ and $\bu\in\mathbb{R}^m$ for all $\bx\in\mathcal{D}\setminus\{\mathbf{0}_n\}$ such that 
        \begin{align}
           L_fV(\bx)+L_gV(\bx)\bu\leq -P(\bx) 
           \label{eqn:lyap_cond}
\end{align}
    \end{itemize}
    \label{def:clf}
\end{definition}
Finally, let the set $\mathcal{K}_{\text{CLF}}(\bx)$ be defined as follows:
\begin{align}
\mathcal{K}_{\text{CLF}}(\bx)=\{\bu\in\mathbb{R}^m|\;L_fV(\bx)+L_gV(\bx)\bu\leq -P(\bx)\}
\label{eqn:kclf_set}
\end{align}
Given a CLF $V$, a stabilizing feedback controller $k(\bx):\mathbb{R}^n\rightarrow \mathbb{R}^m$ can be synthesized by the Universal Sontag’s formula as follows:
\begin{align}
    k(\bx)= \begin{cases}-\frac{L_fV+\sqrt{\left(L_fV\right)^2+\left(L_gV\right)^4}}{L_gV} & L_gV \neq 0 \\ 0 & L_gV=0\end{cases}
    \label{eqn:sontags_formula}
\end{align}
It can be shown that $k(\bx)$ is Lipschitz continuous \cite{sontag1989universal}. Furthermore, if the Small Control Property (SCP) holds \cite{sontag1989universal}, the feedback controller is continuously differentiable and consequently, the existence of a CLF becomes a necessary and sufficient condition for asymptotic stability. If the function $P(\bx)$ is a $\mathcal{K}$ function $\gamma(V(\bx))$ in addition to a positive definite function, then any Lipschitz controller $\bu\in \mathcal{K}_{\text{CLF}}(\bx)$ guarantees exponential stability \cite{freeman2008robust_clf_exp}.
\begin{definition}
    \normalfont \textbf{(CBF)} Consider a scalar valued continuously differentiable function $h:\mathbb{R}^n\rightarrow\mathbb{R}$ such that the following holds true
    \begin{enumerate}
        \item $h(\bx)>0$ for $\bx\in \operatorname{Int}\mathcal{S}$, $h(\bx)<0$ for $\bx\in \mathbb{R}^n\setminus\mathcal{S}$ and $h(\bx)=0$ for $\bx\in \partial \mathcal{S}$
        \item For \eqref{eqn:nonlinear_system_dynamics_affine}, there exist a class $\mathcal{K}_\infty$
        function $\alpha$ such that the following holds true
        \begin{align}
            \underset{\bu\in\mathbb{R}^m}{\sup}\;L_fh(\bx)+L_gh(\bx)\bu\geq -\alpha(h(\bx))
            \label{eqn:cbf_condition}
        \end{align}
    \end{enumerate}
    Then, the function $h$ is said to be a Control Barrier Function (CBF).
\end{definition}
Finally, let the set $\mathcal{K}_{\text{CBF}}(\bx)$ be defined as follows:
\begin{align}
\mathcal{K}_{\text{CBF}}(\bx)=\{\bu\in\mathbb{R}^m|\;L_fh(\bx)+L_gh(\bx)\bu\geq -\alpha(h(\bx))\}
\label{eqn:kcbf_set}
\end{align}
\begin{theorem}
\normalfont \cite{ames2014control_bf_6} Any Lipschitz continuous controller $\bu\in \mathcal{K}_{\text{CBF}}(\bx)$, renders the set $\mathcal{S}$ forward invariant.
\label{thm:cbf}
\end{theorem}
\begin{assumption}
    \normalfont We assume that the system \eqref{eqn:nonlinear_system_dynamics}, there exists a stabilizing and a safe controller.
    \label{assumption:stabilizing_assum_1}
\end{assumption}
\begin{definition}
\normalfont (\textbf{Relative Degree}) A function $h(\bx)$ is said to be of relative degree $m> 1$, if the following holds true
\begin{align}
&L^i_gh(\bx)\bu=0,\;\forall\;\;i\in \mathcal{I}_{(1,m-1)}\;\text{and}\;\forall\bx\in\mathbb{R}^n\nonumber\\
& L^m_gh(\bx)\bu\neq 0,\;\forall\;\;\;\bx\in\mathbb{R}^n\nonumber
\end{align}
\label{defn:relative_degree}
\end{definition}
For systems with a relative degree greater than one, the notion of HO-CBF is presented next.
{
\begin{definition}
    \normalfont \textbf{High Order CBF (HO-CBF)} \cite{xiao2021high_order_cbf_1,tan2021high_order_2}
Consider the system \eqref{eqn:nonlinear_system_dynamics_affine}  with relative degree $m>1$. Let the sequence of functions be defined recursively as $\psi^{i+1}(\boldsymbol{x}) = \dot{\psi}^i(\boldsymbol{x}) + \alpha_i(\psi^i(\boldsymbol{x}))$ for $i\in[0,m-1]_d$, $\psi^0(\boldsymbol{x}) = h(\bx)$ and $\alpha_i$ for $i\in[0,m]_d$ be class $\mathcal{K}_\infty$ functions. We then define a finite collection of sets $C_1(t), C_2(t), \ldots, C_m(t)$ as $C_{i+1}(t) = \{\boldsymbol{x} \in \mathbb{R}^n: \psi^i(\boldsymbol{x}) \geq 0\}$ for $i\in[0,m-1]_d$. A function $h: \mathbb{R}^n \rightarrow \mathbb{R}$ is a HO-CBF, if there exists a differentiable class $\mathcal{K}_\infty$ functions $\alpha_i$ ($i\in[1,m]_d$) such that
    \begin{align}
        \dot{\psi}^m(\bx)\geq -\alpha^m(\psi^m(\bx))\quad\quad\forall \bx\in \cap_{i=1}^mC_i
        \label{eqn:HO-CBF_condition}
    \end{align}
\label{def:high_order_cbf}
\end{definition}
For a given HO-CBF $h$, we define the set of control inputs satisfying the HO-CBF condition \eqref{eqn:HO-CBF_condition} for a given $\bx\in\mathbb{R}^n$:
\begin{align}
\mathcal{K}_{\text{HO-CBF}}(\bx) = \left\{\boldsymbol{u} \in \mathbb{R}^m: \dot{\psi}^m(\bx)\geq -\alpha^m(\psi^m(\bx))\right\}
\end{align}
where $\alpha^m$ is a class $\mathcal{K}_\infty$ function.
\begin{theorem}
    \normalfont \cite{xiao2021high_order_cbf_1,tan2021high_order_2} Let $h(\bx)$ be a HO-CBF $h(\bx)$ for system \eqref{eqn:nonlinear_system_dynamics_affine}. If $\boldsymbol{x}_0\in \cap_{i=1}^mC_i$, then any Lipschitz continuous controller $\boldsymbol{u} \in \mathcal{K}_{\text{HO-CBF}}(\bx)$ ensures that the intersection of these sets remains forward invariant for system \eqref{eqn:nonlinear_system_dynamics}.
\end{theorem}
}

\section{Motivation\label{sec:motivation}}
In this section, we present the motivation for our contributions by examining the general CLF-CBF-QP problem used to synthesize safe and stabilizing feedback controllers:

\begin{subequations}
\begin{align}
\underset{\bu\in\mathbb{R}^m}{\min}\;\;&\|\bu-\bu_0\|^2+p\delta_1^2+q\delta_2^2\\
\text{s.t.}\;\;&L_fV(\bx)+L_gV(\bx)+\gamma(V(\bx))\leq \delta_1\label{eqn:clf_condition}\\
&L_fh(\bx)+L_gh(\bx)+\alpha(h(\bx))\geq \delta_2\label{eqn:cbf_condition_1}
\end{align}
\label{eqn:clf_cbf_qp}
\end{subequations}
where $\bu_0$ is a given nominal controller, $p\geq 0$, $q\geq 0$, and $\delta_1,\delta_2\in\mathbb{R}$ are slack variables that ensures QP feasibility. While this formulation guarantees feasibility, it does not ensure that the synthesized control inputs are safe and stabilizing due to the inclusion of slack variables. Specifically, the stability condition from CLF (Eq. \eqref{eqn:clf_condition} with $\delta_1=0$) and the safety condition from CBF (Eq. \eqref{eqn:cbf_condition_1} with $\delta_2=0$) can be violated when $\delta_1>0$ and $\delta_2<0$, respectively. 
{
The following counterexample shows that there might not exist any Lipschitz controller that solves \eqref{eqn:clf_cbf_qp}. 

\textbf{Robinson's counterexample: \cite{robinson1982generalized_robinson}}
Consider the following QP:
\begin{align}
\mathcal{P}(\bx)= \begin{cases}\underset{\bu \in \mathbb{R}^4}{\operatorname{argmin}} & \frac{1}{2} \bu^\mathrm{T} \bu \\ \text { s.t. } & A(\bx) \bu \geq b(\bx)\end{cases}
\label{eqn:robin_qp}
\end{align}
where $\bx=[x_1, x_2]^\mathrm{T} \in \mathbb{R}^2$ and the matrices $A(\bx)$ and $b(\bx)$ are given by:
$$
A(\bx)=\left[\begin{array}{rrll}
0 & -1 & 1 & 0 \\
0 & 1 & 1 & 0 \\
-1 & 0 & 1 & 0 \\
1 & 0 & 1 & x_1
\end{array}\right],\;\; b(\bx)=\left[\begin{array}{c}
1 \\
1 \\
1 \\
1+x_2
\end{array}\right].
$$
For the domain $0<x_1<1$ and $ 0 \leq x_2 \leq \frac{1}{2} x_1^2$, the unique minimizer to \eqref{eqn:robin_qp} can be expressed in closed form as:
$$
\bu(\bx)=[0,0,1,0]^\mathrm{T}+\frac{x_2}{x_1}[0,0,0,1]^\mathrm{T}
$$
At the point $x_1 = x_2 = 0$, the solution is given by $\bu(\bx) = [0, 0, 1, 0]^\mathrm{T}$. It is noteworthy that, although the existence and uniqueness of the minimizer $\bu(\bx)$ are guaranteed for all values of $\bx = [x_1, x_2]^\mathrm{T}$, the minimizer itself fails to exhibit Lipschitz continuity in any open set containing the origin.
In addition, our contributions are motivated by addressing the following four distinct scenarios, assuming that a CLF and CBF are known apriori. These scenarios highlight the limitations of current approaches and establish the need for our proposed method.

\textit{Case 1:} To ensure that control inputs synthesized via \eqref{eqn:clf_cbf_qp} guarantee both safety and stability, we must set $\delta_1=\delta_2=0$. However, this approach can lead to infeasibility of the QP at certain states. Specifically, the QP becomes infeasible when the intersection of the sets of control inputs satisfying the CLF and CBF conditions is empty, i.e., $\mathcal{K}_{\text{CLF}}(\bx)\cap\mathcal{K}_{\text{CBF}}(\bx)=\varnothing$.

\textit{Case 2}: In \eqref{eqn:clf_cbf_qp}, consider the case where $\delta_1=\delta_2=0$ and there is an additional constraint $\bu\in\mathcal{U}$. At a given $\bx$,
there might not exist a $\bu$ such that \eqref{eqn:clf_cbf_qp} is feasible. In other words, forward invariance of $\mathcal{U}$ is not guaranteed. 
{
Most methods in literature, that consider adding input constraints in the CLF-CBF-QP framework, add a slack variable to the cost function to make the QP feasible but lack guarantees of safety/stability \cite{zeng2021safety_feasibility_input_constraints} or either only consider box constraints for inputs which can be conservative \cite{xiao2022sufficient_cbf}. Furthermore, \cite{xiao2022sufficient_cbf} provides only sufficient conditions for the feasibility of CBF-QP based methods.

}

\textit{Case 3}: If the state constraint specification is given by $\mathcal{S}=\mathcal{S}_1\cup\mathcal{S}_2$, the CLF-CBF-QP can be modified as 
\begin{subequations}
\begin{align}
\underset{\bu\in\mathbb{R}^m}{\min}\;\;&\|\bu-\bu_0\|^2\\
\text{s.t.}\;\;&L_fV(\bx)+L_gV(\bx)+\gamma(V(\bx))\leq 0\\
&L_fh_i(\bx)+L_gh_i(\bx)+\alpha_i(h_i(\bx))\geq 0,\;\;\forall\;\;i\in \mathcal{I}_{(1,2)}\label{eqn:two_cbf_conditions}
\end{align}
\label{eqn:clf_cbf_qp_complex_state_constraints}
\end{subequations}
However, as the state constraint specification is a union of two safe sets, if the QP is feasible both the CBF conditions \eqref{eqn:two_cbf_conditions} are satisfied, which can be conservative. Furthermore, if this specification becomes complex, for instance, $\mathcal{S}=\mathcal{S}_1\cup\mathcal{S}_2\cap\mathcal{S}_3\cup\mathcal{S}_4$, the inclusion of the additional CBF constraints would furthermore result into a more conservative scenario if $\mathcal{S}=\mathcal{S}_1\cup\mathcal{S}_2\cap\mathcal{S}_3\cup\mathcal{S}_4\neq\varnothing$, as the QP tries to guarantee forward invariance for this $\mathcal{S}$. Additionally, if $\mathcal{S}=\mathcal{S}_1\cup\mathcal{S}_2\cap\mathcal{S}_3\cup\mathcal{S}_4=\varnothing$, then the QP in \eqref{eqn:clf_cbf_qp_complex_state_constraints} cannot be used to synthesize safe and stabilizing controllers.

\textit{Case 4}: Finally, we consider the scenario involving nonlinear system \eqref{eqn:nonlinear_system_dynamics}. In this case, the QP formulation \eqref{eqn:clf_cbf_qp} may not be directly applicable, as the CBF condition would introduce nonlinear terms with respect to $\bu$, potentially rendering the QP ineffective.

To summarize, if we have at least one or all of the cases discussed above, using CBF-CLF-QP based methods would not be feasible (Cases 1-3) or not possible (Case 3 when the system \eqref{eqn:nonlinear_system_dynamics} is general).

\section{Main results\label{sec:main_results}}
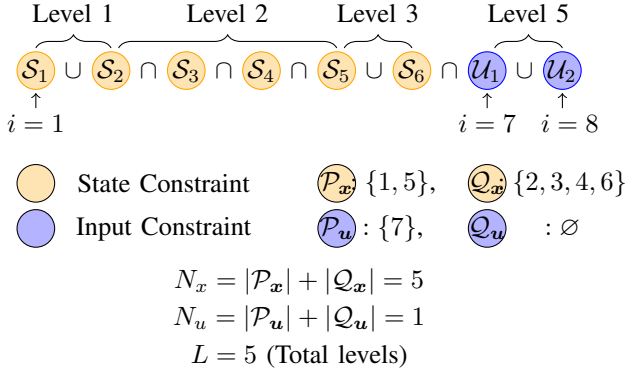
\begin{figure}
    \centering
    \begin{tikzpicture}
\begin{scope}[shift={(1,1)}]
    \definecolor{statecolor}{RGB}{255, 165, 0} 
    \definecolor{inputcolor}{RGB}{0, 0, 255} 

    \filldraw[statecolor, fill opacity=0.3] (1,0) circle (0.25cm);
    \filldraw[statecolor, fill opacity=0.3] (2,0) circle (0.25cm);
    \filldraw[statecolor, fill opacity=0.3] (3,0) circle (0.25cm);
    \filldraw[statecolor, fill opacity=0.3] (4,0) circle (0.25cm);
        \filldraw[statecolor, fill opacity=0.3] (5,0) circle (0.25cm);
            \filldraw[statecolor, fill opacity=0.3] (6,0) circle (0.25cm);

    \filldraw[inputcolor, fill opacity=0.3] (7,0) circle (0.25cm);
    \filldraw[inputcolor, fill opacity=0.3] (8,0) circle (0.25cm);

    \node at (1,0) {$\mathcal{S}_1$};
    \node at (2,0) {$\mathcal{S}_2$};
    \node at (3,0) {$\mathcal{S}_3$};
    \node at (4,0) {$\mathcal{S}_4$};
    \node at (5,0) {$\mathcal{S}_5$};
    \node at (6,0) {$\mathcal{S}_6$};
        \node at (7,0) {$\mathcal{U}_1$};
    \node at (8,0) {$\mathcal{U}_2$};

    \node at (1.5,0) {$\cup$};
    \node at (2.5,0) {$\cap$};
    \node at (3.5,0) {$\cap$};
      \node at (4.5,0) {$\cap$};
        \node at (5.5,0) {$\cup$};
    \node at (6.5,0) {$\cap$};
    \node at (7.5,0) {$\cup$};
         \draw [decorate, decoration={brace, amplitude=5pt}] (1,0.3) -- (2,0.3) 
                node [midway, above=6pt] {Level 1};

                    \draw [decorate, decoration={brace, amplitude=5pt}] (2.1,0.3) -- (5,0.3) 
                node [midway, above=6pt] {Level 2};

                  \draw [decorate, decoration={brace, amplitude=5pt}] (5.1,0.3) -- (6,0.3) 
                node [midway, above=6pt] {Level 3};
                 \draw [decorate, decoration={brace, amplitude=5pt}] (7.1,0.3) -- (8,0.3) 
                node [midway, above=6pt] {Level 5};
    \node[circle, fill=statecolor, fill opacity=0.3, draw=black, minimum width=0.5cm, minimum height=0.5cm] at (1.0, -1.5) {};
    \node at (2.7, -1.5) {State Constraint};

    \node at (4.5, -2.8) {$N_x=|\mathcal{P}_{\bx}|+|\mathcal{Q}_{\bx}|=5$};
        \node at (4.5, -3.3) {$N_u=|\mathcal{P}_{\bu}|+|\mathcal{Q}_{\bu}|=1$};
             \node at (4.5, -3.8) {$L=5$ (Total levels)};

        \node[circle, fill=statecolor, fill opacity=0.3, draw=black, minimum width=0.5cm, minimum height=0.5cm] at (5.0, -1.5) {};
    \node at (5.75, -1.5) {$:\{1,5\}$,};
        \node at (5,-1.5) {$\mathcal{P}_{\bx}$};

                \node[circle, fill=statecolor, fill opacity=0.3, draw=black, minimum width=0.5cm, minimum height=0.5cm] at (7, -1.5) {};
    \node at (8, -1.5) {$:\{2,3,4,6\}$};
        \node at (7,-1.5) {$\mathcal{Q}_{\bx}$};

                \node[circle, fill=inputcolor, fill opacity=0.3, draw=black, minimum width=0.5cm, minimum height=0.5cm] at (5.0, -2.1) {};
    \node at (5.75, -2.1) {$:\{7\}$,};
        \node at (5,-2.1) {$\mathcal{P}_{\bu}$};

                \node[circle, fill=inputcolor, fill opacity=0.3, draw=black, minimum width=0.5cm, minimum height=0.5cm] at (7, -2.1) {};
    \node at (8, -2.1) {$:\varnothing$};
        \node at (7,-2.1) {$\mathcal{Q}_{\bu}$};

    \node[circle, fill=inputcolor, fill opacity=0.3, draw=black, minimum width=0.5cm, minimum height=0.5cm] at (1.0, -2.1) {};
    \node at (2.7, -2.1) {Input Constraint};
    \draw[->] (1,-0.5) -- (1,-0.3);
    \node at (1,-0.7) {$i=1$};

       \draw[->] (7,-0.5) -- (7,-0.3);
    \node at (7,-0.7) {$i=7$};
       \draw[->] (8,-0.5) -- (8,-0.3);
    \node at (8.1,-0.7) {$i=8$};
\end{scope}
\end{tikzpicture}
    \caption{\small 
    Illustration of state and input constraint specifications. The set $\mathcal{P}_{\bx} = \{1\}$ represents a union operation applied after the state constraint set $\mathcal{S}_1$, while $\mathcal{Q}_{\bx} = \{2, 3, 4\}$ represents intersection operations applied after the state constraint sets $\mathcal{S}_2$, $\mathcal{S}_3$, and $\mathcal{S}_4$. Similarly, the set $\mathcal{P}_{\bu} = \{5\}$ represents a union operation applied after the input constraint set $\mathcal{U}_1$, and $\mathcal{Q}_\bu = \varnothing$ indicates that there are no further input constraints following $\mathcal{U}_2$.  Note that the subscript $\bx$ in $\mathcal{P}_{\bx}$ or $\mathcal{Q}_{\bx}$ indicates that there is a union or intersection operation after the state constraint set. Similarly, for $\mathcal{P}_{\bu}$ and $\mathcal{Q}_{\bu}$.
    }
    \label{fig:illus_state_input_spec}
\end{figure}










In Section \ref{subsec:state_and_input_constraints}, we first introduce the nomenclature to describe the complex state and input constraint specification. Next, in Section \ref{subsec:ubf}, we propose the notion of Universal Barrier Functions (UBF) followed by the formulation of UBF based quadratic programs (UBF-QP) for safe and stabilizing control synthesis for input constrained and high order general nonlinear systems subject to input constraints in Section \ref{subsec:ubf_based_quadratic_programs}. 
\subsection{State and Input constraint specifications\label{subsec:state_and_input_constraints}}
Given the individual compact state $\mathcal{S}_i$ and input $\mathcal{U}_i$ constraint sets characterized by $\mathcal{S}_i=\{\bx\in\mathbb{R}^n|\;h_i(\bx)\geq 0\}$ and $\mathcal{U}_i=\{\bu\in\mathbb{R}^m|\;h_i(\bu)\geq 0\}$ respectively, the overall specification denoted by $\mathcal{A}$ is given by \cite{zeng2021safety_feasibility_input_constraints}
\begin{align}
   \mathcal{A}=\mathcal{S}\cap\mathcal{U}
\label{eqn:state_and_input_constraint_set_complex}
\end{align}
where the sets $\mathcal{S}$ and $\mathcal{U}$ are given by
\begin{align}
    &\mathcal{S}=\oplus_{i=1}^{N_x} \{\bx\in\mathbb{R}^n|\;h_i(\bx)\geq 0\},\nonumber\\ &\mathcal{U}=\oplus_{i=1}^{N_u} \{\bu\in\mathbb{R}^m|\;h_i(\bu)\geq 0\}\nonumber
\end{align}
where $\oplus$ is defined as in Section \ref{sec:nomenclature} for some given sets $\pn$ and $\qn$ (based on the state and input constraints specification) and $N=N_x+N_u$. Let $L$ denote the number of levels for the given specification. For instance, $\cup\cap\cap\cap\cup\cup$ and $\cap\cup$ is a three level ($L=3$) and two level ($L=2$) specifications respectively. 
Let $\mathcal{P}'_N\in P(\mathcal{P}_\bx\cup\mathcal{P}_\bu)$ and $\mathcal{Q}'_N\in P(\mathcal{Q}_\bx\cup\mathcal{Q}_\bu)$, where each subset consists of indices that are consecutive numbers. These subsets are the indices of the union operations at a given level $j\in\mathcal{L}_\mathcal{P}$ for $\mathcal{P}'_N$. Similarly, for $\mathcal{Q}'_N$. For instance, for $\cup\cap\cap\cap\cup\cup$, the sets $\mathcal{P}'_N$ and $\mathcal{Q}'_N$ are $\{\{1\},\{5,6\}\}$ and $\{\{2,3,4\}\}$ respectively.   Let $N'_\mathcal{P}$ and $N'_\mathcal{Q}$ denote the number of such subsets in $\mathcal{P}'_N$ and $\mathcal{Q}'_N$ respectively (note that $N'_\mathcal{P}+N'_\mathcal{Q}=L$). In the previous example, $N'_\mathcal{P}=2$ and $N'_\mathcal{Q}=1$. Furthermore, let 
$\mathcal{P}'_N(j)$ denote the subset present at $j^\text{th}$ level where $j\in\mathcal{L}_{\mathcal{P}}$. Similarly, $\mathcal{Q}'_N(j)$ denote the subset present at $j^\text{th}$ level where $j\in\mathcal{L}_{\mathcal{Q}}$.
\begin{assumption}
    \normalfont We assume that the sets  $\mathcal{S}_i$ (for $i\in\mathcal{I}_\bx$) and $\mathcal{U}_i$ (for $i\in\mathcal{I}_\bu$) are compact and $\mathcal{A}$ is non-empty. This is a reasonable assumption to make from a practical perspective.
    \label{assumption:assum_2}
\end{assumption}
{
\begin{assumption}
    \normalfont We assume that $(\mathbb{R}^n\times\mathbb{R}^m)\setminus\mathcal{A}$ is unbounded. The motivation for this assumption is that, under only safe state constraints $\mathcal{S}$, if $\mathbb{R}^n\setminus\mathcal{S}$ (where $\mathcal{S}$ is a closed safe set) is bounded, from Theorem 11 in \cite{braun2017existence}, there does not exist a locally Lipschitz safe stabilizing controller.
        \label{assumption:assum_3}
\end{assumption}
}
\begin{example}
    \normalfont 
    Consider the state and input constraint specification given by
\begin{align}
\mathcal{A}=\mathcal{S}_1\cup\mathcal{S}_2\cap\mathcal{S}_3\cap\mathcal{S}_4\cap\mathcal{S}_5\cup\mathcal{S}_6\cap\mathcal{U}_1\cup\mathcal{U}_2
\label{eqn:ubf_specification_example_1}
\end{align}
In this specification, there are five levels i.e, $L=5$ (in particular, $\mathcal{S}_1\cup\mathcal{S}_2$, $\mathcal{S}_2\cap\mathcal{S}_3\cap\mathcal{S}_4\cap\mathcal{S}_5$, 
 $\mathcal{S}_5\cup\mathcal{S}_6$, $\mathcal{S}_6\cap\mathcal{U}_1$, and $\mathcal{U}_1\cup\mathcal{U}_2$). The sets $\pn=\{1,5\}$ and $\qn=\{2,3,4,6\}$. 
    The set $\mathcal{P}'_N$ consists of sequence of subsets where each subset consists of indices at a particular level that have union operations i.e. $\mathcal{P}'_N=\{\{1\},\{5\},\{7\}\}\}$. Similarly, $\mathcal{Q}'_N=\{\{\{2,3,4\},\{6\}\}$. Furthermore, $\mathcal{P}'_N(2)=\{5\}$ as, at level 3, we have union operations at index 5. Similarly, $\mathcal{Q}'_N(1)=\{2,3,4\}$ and $\mathcal{Q}'_N(2)=\{6\}$.
\label{example:state_and_input_constraint_specification}
\end{example}
\subsection{Union and Intersection operations}
Note that the union of sets $\mathcal{B}=\underset{i\in\mathcal{I}_N}{\cup}\;\mathcal{B}_i$ can be equivalently written as
\begin{align}
    \left\{\bx|\;\exists\;i\in \iin\;\;h_i(\bx)\geq 0\right\}\logeq\left\{\bx|\;\underset{i\in \iin}{\max}\;h_i(\bx)\geq 0\right\}
    \label{eqn:union_equivalent_condition}
\end{align}
Similarly, the intersection of sets $\underset{i\in\mathcal{I}_N}{\cap}\;\mathcal{B}_i=\{\bx|\;h_i(\bx)\geq 0,\;i\in \iin\}$ can be equivalently written as
\begin{align}
    \left\{\bx|\;h_i(\bx)\geq 0,\;\forall\;\;i\in \iin\right\}\logeq\left\{\bx|\;\underset{i\in \iin}{\min}\;h_i(\bx)\geq 0\right\}
\label{eqn:intersection_equivalent_condition}
\end{align}
With a slight abuse of notation, the inputs to $h_i$ are either $\bx$ or $\bu$. For $i\in\iin$, the appropriate inputs to $h_i$ can be inferred from the context.
\subsection{Lyapunov-based stability condition as a forward invariance condition}
Before we proceed with presenting our notion of UBF, we present a methodology to view the Lyapunov-based stability condition \eqref{eqn:lyap_cond} as a forward invariance condition on a set. Towards this goal, let $V$ be a CLF (Definition \ref{def:clf}) and consider the following set
\begin{align}
    \mathcal{K}_V=\{(\bx,\bu)\in\mathbb{R}^n\times\mathbb{R}^m:\;\dot{V}(\bx)\leq -P(\bx)\}
    \label{eqn:kv}
\end{align}
Define $h_V(\bx,\bu)=-\dot{V}(\bx)-P(\bx)$ and let $\mathcal{S}_V\subset\mathcal{K}_V:=\{(\bx,\bu)\in\mathbb{R}^n\times\mathbb{R}^m|\;h_V(\bx,\bu)\geq 0\}$ be a compact set. Note that the set $\mathcal{K}_V\subset\mathcal{D}$ where $\mathcal{D}$ is an open set defined in Definition \ref{def:clf}. Then, if there exists a $\mathcal{K}_\infty$ function $\alpha$ such that the following holds true
\begin{align}
    \dot{h}_V(\bx,\bu)\geq -\alpha(h_V(\bx,\bu))\quad\forall (\bx,\bu)\in\mathcal{S}_V
    \label{eqn:lyap_as_forward_invariance}
\end{align}
Then, for any $(\bx_0,\bu_0)\in\mathcal{S}_V$, it follows from Nagumo theorem for forward invariance, $(\bx(t),\bu(t))\in\mathcal{S}_V$ ($t\geq 0)$ which implies that $\bu(t)$ that satisfies \eqref{eqn:lyap_as_forward_invariance} is a stabilizing controller. Consequently, the specification characterized by the set $\mathcal{A}$ defined in \eqref{eqn:state_and_input_constraint_set_complex} is modified to include the stability specification as follows:
\begin{align}
\mathcal{A}_s=\mathcal{A}\cap\mathcal{S}_V
    \label{eqn:specification_with_stability}
\end{align}
It must be noted that in the expression of $\dot{h}_V(\bx,\bu)$, an integral control term is present. The reason for presenting this view is mainly to encode the stability and safety conditions into a single scalar-valued continuously differentiable function. This will become clearer in subsequent sections.
\begin{figure}
\begin{tikzpicture}
    \input{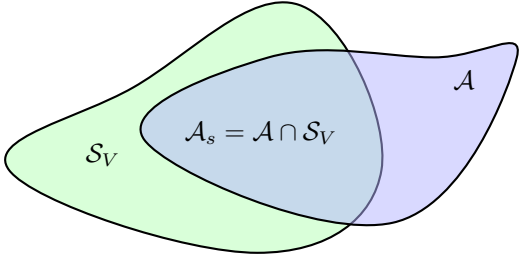}
\end{tikzpicture}
\caption{The depiction of sets $\mathcal{A}_s$ \eqref{eqn:specification_with_stability}, $\mathcal{A}$\eqref{eqn:state_and_input_constraint_set_complex}, and $\mathcal{S}_V$ characterized by the CLF $V$ condition \eqref{eqn:clf_condition}. }
\label{fig:sets}
\end{figure}
\subsection{Universal Barrier Functions (UBFs)\label{subsec:ubf}}
In this section, 
we present a scalar continuously differentiable function termed UBF that can simultaneously encode the notions of stability and safety, handle complex state and input constraints specifications, and be applicable to general nonlinear systems \eqref{eqn:nonlinear_system_dynamics} and systems with higher relative degrees.
\begin{definition}
    \normalfont \textbf{(UBF)} Given the system \eqref{eqn:nonlinear_system_dynamics}, the set $\mathcal{A}_s$ in \eqref{eqn:state_and_input_constraint_set_complex} characterized by continuously differentiable functions $h_i(\bx)$ for state constraints and $h_i(\bu)$ for input constraints and a CLF $V$, consider a scalar valued continuously differentiable function $h:\mathbb{R}^n\times\mathbb{R}^m\rightarrow\mathbb{R}$ defined as follows:
    \begin{align}
&h(\bx,\bu)=c_2\left(H_{N+1}(\bx,\bu)  \right)+c_2(b_N)
\label{eqn:ubf}
    \end{align}
    where 
     $\beta>0$, $b_N=\left({\prod_{i=1}^{|\mathcal{P}'_N|}\;|\mathcal{P}'_N(i)|}\right)^{-1}$, and $H_{N+1}(\bx,\bu)$ is recursively defined by
    \begin{subequations}
    \begin{align}
    & H_0(\bx,\bu)=\left\{\begin{array}{ll}
     &c_1(-\beta (\dot{V}(\bx)-P(\bx))), \quad \text{if } i=N+1  \\
     & c_1(\beta h_i(\bx)), \quad \quad\quad\quad\quad\quad\text{if } i\in \ix\\
     &c_1(\beta h_i(\bu)), \quad\quad\quad\quad\quad\quad \text{if } i\in \iu
\end{array}\right.\\
&H_{j+1}(\bx,\bu)\nonumber\\
&=\left\{\begin{array}{ll} 
H_j(\bx,\bu)+c_1(\dot{V}(\bx)), \quad \quad\quad\quad\quad\quad\quad\quad j=N+1\\
H_j(\bx,\bu)+c_1(h_{j+1}(\bx)), \quad \quad\quad\quad\quad\quad\quad j\in \pn\\
\left((H_j(\bx,\bu))^{-1}+(c_1(h_j(\bx)))^{-1}\right)^{-1}, \quad\quad j\in \qn\\
H_j(\bx,\bu)+c_1(h_{j+1}(\bu)), \quad \quad\quad\quad\quad\quad\quad j\in \pu\\
\left((H_j(\bx,\bu))^{-1}+(c_1(h_{j+1}(\bu)))^{-1}\right)^{-1}, \quad j\in \qu\\
\end{array}\right.
\end{align}
\label{eqn:ubf_math_equation}
    \end{subequations}
where $c_1(r)=\mathrm{e}^{\beta r}$, $c_2(r)=\text{ln($r$)}/\beta$ for $r\in\mathbb{R}$. Then, $h$ is said to be a Universal Barrier Function (UBF), if there exists a $\bu\in\mathbb{R}^m$ and class $\mathcal{K}_\infty$ function $\alpha$ such that the following condition holds true
        \begin{align}
&\dot{h}(\bx,\bu)+\alpha(h(\bx,\bu))\geq 0
    \label{eqn:ubf_condition}
    \end{align}
   for every $\bx\in\mathbb{R}^n$.
   \label{defn:ubf}
\end{definition}

\begin{table}[h]
    \centering
    \begin{tabular}{|c|c|}
    \hline
         Functions ($\beta>0$)  & $S_\beta(x_1,x_2)$  \\ \hline
         Boltzmann operator & $\frac{x_1\mathrm{e}^{\beta x_1}+x_2\mathrm{e}^{\beta x_2}}{\mathrm{e}^{\beta x_1}+\mathrm{e}^{\beta x_2}}$\\ \hline
    Log-Sum-Exp & $\frac{1}{\beta}\text{log}\left(\mathrm{e}^{\beta x_1}+\mathrm{e}^{\beta x_2}\right)$ \\ \hline
    Mellowmax & $\frac{1}{\beta}\text{log}\left(\frac{\mathrm{e}^{\beta x_1}+\mathrm{e}^{\beta x_2}}{2}\right)$  \\ \hline
    $\beta$-Norm & $\left(x_1^\beta+x_2^\beta\right)^{\frac{1}{\beta}}$  \\ \hline
    \end{tabular}
    \caption{Candidate functions ${S}_\beta(x_1,x_2)$}
    \label{tab:functions}
\end{table}

Let $h(\bx,\bu)$ be a UBF, then we define $\mathcal{K}_{\text{UBF}}(\bx)$ as follows:
\begin{align}
    \mathcal{K}_{\text{UBF}}(\bx)=\{\bu\in\mathbb{R}^m|\; \dot{h}(\bx,\bu)+\alpha(h(\bx,\bu))\geq 0\}
    \label{eqn:k_ubf}
\end{align}
{
\begin{lemma}
\normalfont
For a parameter $\beta>0$, let $ S_\beta(x_1, x_2) $ and $ S_{-\beta}(x_1, x_2) $  be classes of functions from $\mathbb{R}^2$ to $\mathbb{R}$ that are smooth approximations of the $ \max\{x_1, x_2\} $ and $ \min\{x_1, x_2\} $ functions respectively, such that
\begin{align}
   & \left\{ (x_1, x_2) \,\bigg|\, \lim_{\beta \to \infty} S_\beta(x_1, x_2) \geq 0 \right\} \nonumber\\
   &= \left\{ (x_1, x_2) \,\bigg|\, \max\{x_1, x_2\} \geq 0 \right\}.\nonumber\\
   &\left\{ (x_1, x_2) \,\bigg|\, \lim_{\beta \to \infty} S_{-\beta}(x_1, x_2) \geq 0 \right\} \nonumber\\
&= \left\{ (x_1, x_2) \,\bigg|\, \min\{x_1, x_2\} \geq 0 \right\}.\nonumber
\end{align}
where $x_1,\;x_2\in\mathbb{R}$. Then, the set
   \begin{align}
      & \left\{ (x_1, x_2, x_3) \,\bigg|\, \lim_{\beta \to \infty} S_\beta\big(S_\beta(x_1, x_2), x_3\big) \geq 0 \right\}\nonumber\\
       &=\left\{ (x_1, x_2, x_3) \,\bigg|\, \max\{x_1, x_2, x_3\} \geq 0 \right\}.\nonumber
      \end{align}
where $x_3\in\mathbb{R}$. Similarly, the set
   \begin{align}
      & \left\{ (x_1, x_2, x_3) \,\bigg|\, \lim_{\beta \to \infty} S_{-\beta}\big(S_{-\beta}(x_1, x_2), x_3\big) \geq 0 \right\}\nonumber\\
        &=  \left\{ (x_1, x_2, x_3) \,\bigg|\, \min\{x_1, x_2, x_3\} \geq 0 \right\}.\nonumber
      \end{align}
\end{lemma}
\begin{proof}
   Since $ S_\beta(x_1, x_2) $ smoothly approximates $ \max\{x_1, x_2\} $, we have $\underset{\beta\to\infty}{\lim} S_\beta(x_1, x_2) = \max\{x_1, x_2\}$
   for all $ x_1, x_2\in\mathbb{R} $. We denote by $S_\infty:=\underset{\beta\to\infty}{\lim} S_\beta(x_1, x_2) $.
   Now, consider the nested function $ S_\beta\big(S_\beta(x_1, x_2), x_3\big) $. Taking the limit as $ \beta \to \infty $, we get
   \begin{align*}
   \lim_{\beta \to \infty} S_\beta\big(S_\beta(x_1, x_2), x_3\big) &= S_\infty\big(S_\infty(x_1, x_2), x_3\big) \\
   &= \max\big\{\max\{x_1, x_2\}, x_3\big\} \\
   &= \max\{x_1, x_2, x_3\}.
   \end{align*}
Therefore, the set where this limit is non-negative is
   \begin{align}
   &\left\{ (x_1, x_2, x_3) \,\bigg|\, \lim_{\beta \to \infty} S_\beta\big(S_\beta(x_1, x_2), x_3\big) \geq 0 \right\}\nonumber\\
   &= \left\{ (x_1, x_2, x_3) \,\bigg|\, \max\{x_1, x_2, x_3\} \geq 0 \right\}.
   \end{align}
   Similarly, since $ S_{-\beta}(x_1, x_2) $ smoothly approximates $ \min\{x_1, x_2\} $, we have $   \lim_{\beta \to \infty} S_{-\beta}(x_1, x_2) = \min\{x_1, x_2\}$
   for all $ x_1, x_2 $.
   Now, consider the nested function $ S_{-\beta}\big(S_{-\beta}(x_1, x_2), x_3\big) $. Taking the limit as $ \beta \to \infty $, we get
   \begin{align*}
   \lim_{\beta \to \infty} S_{-\beta}\big(S_{-\beta}(x_1, x_2), x_3\big) &= S_{-\infty}\big(S_{-\infty}(x_1, x_2), x_3\big) \\
   &= \min\big\{\min\{x_1, x_2\}, x_3\big\} \\
   &= \min\{x_1, x_2, x_3\}.
   \end{align*}
   Therefore, the set where this limit is non-negative is
   \begin{align}
  & \left\{ (x_1, x_2, x_3) \,\bigg|\, \lim_{\beta \to \infty} S_{-\beta}\big(S_{-\beta}(x_1, x_2), x_3\big) \geq 0 \right\}\nonumber\\
  &= \left\{ (x_1, x_2, x_3) \,\bigg|\, \min\{x_1, x_2, x_3\} \geq 0 \right\}.
   \end{align}
   \end{proof}
\begin{remark}
   \normalfont  By sequentially applying the smooth approximations and taking the limit as $ \beta \to \infty $, we recover the $ \max $ and $ \min $ functions over multiple variables. Thus, the sets defined by the limits of these nested functions correspond precisely to the sets where the maximum or minimum of the variables is non-negative. 
\end{remark}
\begin{remark}
    \normalfont Note that other smooth functions used to approximate the $\max$ and $\min$ operators for union and intersection sets respectively such as Boltzmann functions, Mellowmax functions, softmax functions, $\beta$-Norm functions, Smooth Maximum Unit (SMU) etc. (see Table \ref{tab:functions}), can be used in place of log-sum-exp function in Definition \ref{defn:ubf}. However, we restrict ourselves to log-sum-exp (LSE) expressions in this paper.
\end{remark}
}
\begin{figure}
\begin{tikzpicture}
    \input{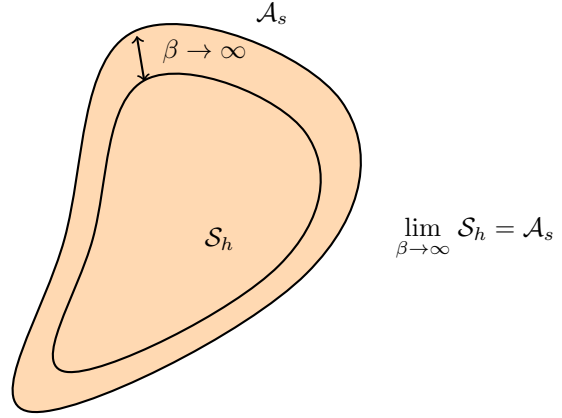}
\end{tikzpicture}
\caption{The set $\mathcal{S}_h$ characterized by the UBF $h$ provides an inner approximation for the safe set $\mathcal{A}_s$.}
\label{fig:sets}
\end{figure}

As we show in the subsequent theorem, the set characterized by the UBF $h(\bx,\bu)$ i.e, $\mathcal{S}_h=\{(\bx,\bu)\in\mathbb{R}^n\times\mathbb{R}^m|\;h(\bx,\bu)\geq 0\}$ over approximates the set $\mathcal{A}_s$.
{

\begin{lemma}
    \label{lemma:union}
    \normalfont
    (Union of constraint sets) Consider the specification given by the union of sets i.e., $ \mathcal{A}_s = \bigcup_{i=1}^N \mathcal{S}_i$
    where each $\mathcal{S}_i = \{ (\bx, \bu) \in\mathbb{R}^n\times\mathbb{R}^m\mid h_i(\bx, \bu) \geq 0 \}$. Then, the set $\mathcal{A}_s$ is a superset of $\mathcal{S}_h = \{ (\bx, \bu) \in\mathbb{R}^n\times\mathbb{R}^m\mid h(\bx, \bu) \geq 0 \}$, where
    \begin{align}
    h(\bx, \bu) = \frac{1}{\beta} \ln\left( \sum_{i=1}^N \mathrm{e}^{\beta h_i(\bx, \bu)} \right) - \frac{\ln N}{\beta}.
    \end{align}
    Furthermore, as $\beta \rightarrow \infty$, $\mathcal{S}_h$ converges to $\mathcal{A}_s$, i.e., $  \underset{\beta \rightarrow \infty}{\lim}\; \mathcal{S}_h = \mathcal{A}_s$.
\end{lemma}

\begin{proof}
    Consider the following inequality:
    \begin{align}
    \mathrm{e}^{\beta \max_{i} h_i(\bx, \bu)} \leq \sum_{i=1}^N \mathrm{e}^{\beta h_i(\bx, \bu)} \leq N \mathrm{e}^{\beta \max_{i} h_i(\bx, \bu)}.
    \end{align}
    Taking the natural logarithm and dividing by $\beta$:
    \begin{align}
    \max_{i} h_i(\bx, \bu) & \leq \frac{1}{\beta} \ln\left( \sum_{i=1}^N \mathrm{e}^{\beta h_i(\bx, \bu)} \right)\nonumber\\
    &\leq \max_{i} h_i(\bx, \bu) + \frac{\ln N}{\beta}.
    \label{eqn:max_less}
    \end{align}
    Subtracting $\frac{\ln N}{\beta}$ from the middle term of \eqref{eqn:max_less} gives,
    \begin{align}
    h(\bx, \bu) = \frac{1}{\beta} \ln\left( \sum_{i=1}^N \mathrm{e}^{\beta h_i(\bx, \bu)} \right) - \frac{\ln N}{\beta} \leq \max_{i} h_i(\bx, \bu).
    \end{align}
    Therefore, if $h(\bx, \bu) \geq 0$, i.e., $(\bx,\bu)\in\mathcal{S}_h$, then $\max_{i} h_i(\bx, \bu) \geq 0$, implying that $(\bx, \bu) \in \mathcal{A}_s$. Hence, $\mathcal{S}_h \subseteq \mathcal{A}_s$.
    
    Furthermore, as $\beta \rightarrow \infty$, the term $\frac{\ln N}{\beta} \rightarrow 0$, and
    \begin{align}
    \lim_{\beta \rightarrow \infty} h(\bx, \bu) = \lim_{\beta \rightarrow \infty} \frac{1}{\beta} \ln\left( \sum_{i=1}^N \mathrm{e}^{\beta h_i(\bx, \bu)} \right) = \max_{i} h_i(\bx, \bu).
    \end{align}
    Thus, $\underset{\beta \rightarrow \infty}{\lim}\; \mathcal{S}_h = \mathcal{A}_s$.
\end{proof}

\begin{lemma}
    \label{lemma:intersection}
    \normalfont
    (Intersection of constraint sets) Consider the specification given by the intersection of sets i.e., $\mathcal{A}_s = \bigcap_{i=1}^N \mathcal{S}_i$
    where each $\mathcal{S}_i = \{ (\bx, \bu) \in\mathbb{R}^n\times\mathbb{R}^m\mid h_i(\bx, \bu) \geq 0 \}$. Then, the set $\mathcal{A}_s$ is a superset of $\mathcal{S}_h = \{ (\bx, \bu) \in\mathbb{R}^n\times\mathbb{R}^m\mid h(\bx, \bu) \geq 0 \}$, where
    \begin{align}
    h(\bx, \bu) = -\frac{1}{\beta} \ln\left( \sum_{i=1}^N \mathrm{e}^{-\beta h_i(\bx, \bu)} \right) .
    \end{align}
    Furthermore, as $\beta \rightarrow \infty$, $\mathcal{S}_h$ converges to $\mathcal{A}_s$, i.e., $  \underset{\beta \rightarrow \infty}{\lim}\; \mathcal{S}_h = \mathcal{A}_s$.
\end{lemma}

\begin{proof}
    Consider the following inequality,
    \begin{align}
    \mathrm{e}^{-\beta \min_{i} h_i(\bx, \bu)} \leq \sum_{i=1}^N \mathrm{e}^{-\beta h_i(\bx, \bu)} \leq N \mathrm{e}^{-\beta \min_{i} h_i(\bx, \bu)}.
    \end{align}
    Taking the natural logarithm and dividing by $-\beta$:
    \begin{align}
        \min_{i} h_i(\bx, \bu) \geq & -\frac{1}{\beta} \ln\left( \sum_{i=1}^N \mathrm{e}^{-\beta h_i(\bx, \bu)} \right) \nonumber\\
        \geq & \min_{i} h_i(\bx, \bu) - \frac{\ln N}{\beta}.
        \end{align}
    Therefore, if $h(\bx, \bu) \geq 0$, then $\min_{i} h_i(\bx, \bu) \geq 0$, implying that $(\bx, \bu) \in \mathcal{A}_s$. Hence, $\mathcal{S}_h \subseteq \mathcal{A}_s$.
    
    Furthermore, as $\beta \rightarrow \infty$, the term $\frac{\ln N}{\beta} \rightarrow 0$, and
    \begin{align}
        \lim_{\beta \rightarrow \infty} h(\bx, \bu) = &\lim_{\beta \rightarrow \infty} -\frac{1}{\beta} \ln\left( \sum_{i=1}^N \mathrm{e}^{-\beta h_i(\bx, \bu)} \right)\nonumber\\
    =& \min_{i} h_i(\bx, \bu).
        \end{align}
    Thus, $\underset{\beta \rightarrow \infty}{\lim}\; \mathcal{S}_h = \mathcal{A}_s$. Hence proved.
\end{proof}
\begin{remark}
   \normalfont Note that the set $\mathcal{S}_h$ characterized by the UBF $h$ must be an under-approximation of the safe set $\mathcal{A}_s$. This ensures that if the state input pair lies within $\mathcal{S}_h$, it must also lie within $\mathcal{A}_s$.
\end{remark}
\begin{theorem}
    \normalfont
    (General Case) Consider the general state and input constraint specification $\mathcal{A}_s$, possibly involving multiple union and intersection operations and a UBF $h$. Then, the set $\mathcal{A}_s$ is a superset of the set $\mathcal{S}_h = \{ (\bx, \bu) \in\mathbb{R}^n\times\mathbb{R}^m\mid h(\bx, \bu) \geq 0 \}$. Furthermore, as $\beta \rightarrow \infty$, $\mathcal{S}_h$ converges to $\mathcal{A}_s$, i.e., $    \underset{\beta \rightarrow \infty}{\lim}\; \mathcal{S}_h = \mathcal{A}_s$.
    \label{thm:subset}
\end{theorem}

\begin{proof}
    We will prove this theorem by mathematical induction on the number of operations (levels) in the specification.
For $N = 1$, the specification $\mathcal{A}_s$ consists of a single set $\mathcal{S}_1$, and $h(\bx, \bu) = h_1(\bx, \bu)$. The theorem trivially holds in this case. Define, for $n<N$, the region $\mathcal{A}_s^n$ taken up to the first $n$ levels of specification $\mathcal{A}_s$. Consequently, the set $\mathcal{S}^n_h=\{(\bx,\bu)|\;h^n(\bx,\bu)\geq 0\}$ where $h^n(\bx,\bu)$ (for $n<N$) is defined as 
\begin{align}
    h^n(\bx,\bu)=c_2\left(H_n(\bx,\bu)  \right)+c_2(b_n)
    \label{eqn:h_up_n}
\end{align}
 Now, assume that for general $n<N$, $\underset{\beta\rightarrow\infty}{\lim}\;\;\mathcal{S}^n_h=\mathcal{A}_s^n$ and $\mathcal{A}_s^n\supseteq\mathcal{S}_h^n$ holds true. The task is to then prove that $\underset{\beta\rightarrow\infty}{\lim}\;\;\mathcal{S}^{n+1}_h=\mathcal{A}_s^{n+1}$ and $\mathcal{A}_s^{n+1}\supseteq\mathcal{S}_h^{n+1}$. While moving from step $n$ to step $n+1$, there is either a union operation (i.e., Case 1) or there is an intersection operation (i.e., Case 2). Accordingly, we consider the following two subcases:

\textbf{Case 1:} (Union operation at $n+1$ level)
At the $n+1$ level, suppose the operation is a union. That is, the specification up to level $n+1$ is:
\begin{align}
    \mathcal{A}_s^{n+1} = \mathcal{A}_s^n \cup \mathcal{B}
\end{align}
where $\mathcal{B}$ is a set defined at level $n+1$ with corresponding function $h_{n+1}(\bx,\bu)$. The corresponding function $h^{n+1}(\bx,\bu)$ is then defined using \eqref{eqn:h_up_n}. Since the union operation corresponds to the smooth maximum, we can define:
\begin{align}
    h^{n+1}(\bx,\bu) = &c_2(H_{n+1}(\bx,\bu))+c_2(b_{n+1})\nonumber
\end{align}
Now, $H_{n+1}(\bx,\bu)=H_n(\bx,\bu)+c_1(h_{n+1}(\bx,\bu))$ and using the fact that $H_n(\bx,\bu)=c_2^{-1}\left(h^n(\bx,\bu)-c_2(b_n)\right)$, we have
\begin{align}
      &h^{n+1}(\bx,\bu)\nonumber\\
      &=c_2\left( c_2^{-1}\left(h^n(\bx,\bu)-c_2(b_n)\right)+c_1(h_{n+1}(\bx,\bu))  \right)\nonumber\\&\;\;\;\; + c_2(b_{n+1})  
\end{align}
Using the fact that $c_2^{-1}(r)=c_1(r)$ and $c_1(r_1+r_2)=c_1(r_1)c_1(r_2)$, we have
\begin{align}
         &h^{n+1}(\bx,\bu)\nonumber\\
         &=c_2\left( c_1\left(h^n(\bx,\bu)-c_2(b_n)\right)+c_1(h_{n+1}(\bx,\bu))  \right) + c_2(b_{n+1}) \nonumber\\
                  &=c_2\left( c_1\left(h^n(\bx,\bu)\right)c_1\left(c_2(b_n)\right)+c_1(h_{n+1}(\bx,\bu))  \right) + c_2(b_{n+1}) \nonumber\\
         &=c_2\left( c_1\left(h^n(\bx,\bu)\right)+c_1(h_{n+1}(\bx,\bu))  \right) + c_2(b_{n+1})\nonumber
\end{align}
Substituting $c_1(r) = \mathrm{e}^{\beta r}$ and $c_2(r) = \frac{\ln(r)}{\beta}$, we get:
\begin{align}
    h^{n+1}(\bx,\bu)
   = \frac{1}{\beta} \ln\left( \mathrm{e}^{\beta h^n(\bx,\bu)} + \mathrm{e}^{\beta h_{n+1}(\bx,\bu)} \right) + \frac{\ln b_{n+1}}{\beta}
   \label{eqn:h_n_plus_one}
\end{align}
Our goal is to show that $\mathcal{A}_s^{n+1} \supseteq \mathcal{S}_h^{n+1}$ and that $\underset{\beta\rightarrow\infty}{\lim}\;h^{n+1}(\bx,\bu) = \underset{}{\max}\left\{ h^n(\bx,\bu), h_{n+1}(\bx,\bu) \right\}$.
Consider the following inequality:
\begin{align}
    \mathrm{e}^{\beta \underset{}{\max}\left\{ h^n(\bx,\bu), h_{n+1}(\bx,\bu) \right\}} &\leq \mathrm{e}^{\beta h^n(\bx,\bu)} + \mathrm{e}^{\beta h_{n+1}(\bx,\bu)} \nonumber\\
    &\leq 2 \mathrm{e}^{\beta \underset{}{\max}\left\{ h^n(\bx,\bu), h_{n+1}(\bx,\bu) \right\}}
    \label{eqn:subcase3a_inequality}
\end{align}
Taking the logarithm on both sides and dividing by $\beta$, we obtain:
\begin{align}
   & \underset{}{\max}\left\{ h^n(\bx,\bu), h_{n+1}(\bx,\bu) \right\} \nonumber\\
    &\leq \frac{1}{\beta} \ln\left( \mathrm{e}^{\beta h^n(\bx,\bu)} + \mathrm{e}^{\beta h_{n+1}(\bx,\bu)} \right)\nonumber\\
    &\leq \underset{}{\max}\left\{ h^n(\bx,\bu), h_{n+1}(\bx,\bu) \right\} + \frac{\ln 2}{\beta}
\end{align}
which implies
\begin{align}
    &\frac{1}{\beta} \ln\left( \mathrm{e}^{\beta h^n(\bx,\bu)} + \mathrm{e}^{\beta h_{n+1}(\bx,\bu)} \right) - \frac{\ln 2}{\beta} \nonumber\\
    &\leq \underset{}{\max}\left\{ h^n(\bx,\bu), h_{n+1}(\bx,\bu) \right\}
\end{align}
Using \eqref{eqn:h_n_plus_one} and the fact that $b_{n+1}\leq \frac{1}{2}$, we have
\begin{align}
    h^{n+1}(\bx,\bu) \leq \underset{}{\max}\left\{ h^n(\bx,\bu), h_{n+1}(\bx,\bu) \right\}
\end{align}
Thus, if $h^{n+1}(\bx,\bu) \geq 0$, then $\underset{}{\max}\left\{ h^n(\bx,\bu), h_{n+1}(\bx,\bu) \right\} \geq 0$. This implies that $(\bx,\bu) \in \mathcal{A}_s^{n+1}$. Therefore, $\mathcal{A}_s^{n+1} \supseteq \mathcal{S}_h^{n+1}$.
Taking the limit as $\beta \rightarrow \infty$ on both sides of \eqref{eqn:subcase3a_inequality}, we get:
\begin{align}
    \underset{\beta\rightarrow\infty}{\lim}\;h^{n+1}(\bx,\bu) = \underset{}{\max}\left\{ h^n(\bx,\bu), h_{n+1}(\bx,\bu) \right\}
\end{align}
This shows that $\underset{\beta\rightarrow\infty}{\lim}\;\mathcal{S}_h^{n+1} = \mathcal{A}_s^{n+1}$.

\textbf{Case 2:} (Intersection operation at $n+1$ level)
At the $n+1$ level, suppose the operation is an intersection. The specification up to level $n+1$ is:
\begin{align}
    \mathcal{A}_s^{n+1} = \mathcal{A}_s^n \cap \mathcal{B}
\end{align}
where $\mathcal{B}$ is a set defined at level $n+1$ with corresponding function $h_{n+1}(\bx,\bu)$. The corresponding function $h^{n+1}(\bx,\bu)$ is then defined using the smooth approximation of the intersection operation.
Since the intersection operation corresponds to the smooth minimum, we define:
\begin{align}
    h^{n+1}(\bx,\bu) = &c_2(H_{n+1}(\bx,\bu))+c_2(b_{n+1})\nonumber
\end{align}
Now, $H_{n+1}(\bx,\bu)=\left(H^{-1}_n(\bx,\bu)+c^{-1}_1(h_{n+1}(\bx,\bu))\right)^{-1}$ and using the facts that $c_1^{-1}(r)=c_1(-r)$ and $c_2(r^{-1})=-c_2(r)$, we have
\begin{align}
    h^{n+1}(\bx,\bu)=&-c_2\left(H^{-1}_n(\bx,\bu)+c_1(-h_{n+1}(\bx,\bu))\right)\nonumber\\
    &+c_2(b_{n+1})
\end{align}
Using the fact that $H^{-1}_n(\bx,\bu)=c_2^{-1}(-h^n(\bx,\bu)+c_2(b_n)$ and $c_2^{-1}(r)=c_1(r)$, we have
\begin{align}
     &h^{n+1}(\bx,\bu)\nonumber\\
     &=-c_2\left(c_1(-h^n(\bx,\bu)+c_2(b_n))+c_1(-h_{n+1}(\bx,\bu))\right)\nonumber\\
     &+c_2(b_{n+1}) \nonumber\\
    &=-c_2\left(c_1(-h^n(\bx,\bu))+c_1(-h_{n+1}(\bx,\bu))\right)+c_2(b_{n+1}) \nonumber
\end{align}
 Substituting $c_1(r) = \mathrm{e}^{\beta r}$ and $c_2(r) = \frac{\ln(r)}{\beta}$, we get:
\begin{align}
    h^{n+1}(\bx,\bu) =& -\frac{1}{\beta} \ln\left( \mathrm{e}^{-\beta h^n(\bx,\bu)} + \mathrm{e}^{-\beta h_{n+1}(\bx,\bu)} \right) \nonumber\\
    &+\frac{\ln b_{n+1}}{\beta}
\end{align}
Our goal is to show that $\mathcal{A}_s^{n+1} \supseteq \mathcal{S}_h^{n+1}$ and that $\underset{\beta\rightarrow\infty}{\lim}\;h^{n+1}(\bx,\bu) = \underset{}{\min}\left\{ h^n(\bx,\bu), h_{n+1}(\bx,\bu) \right\}$. Towards this goal, we consider the following inequality:
\begin{align}
    \mathrm{e}^{-\beta \underset{}{\min}\left\{ h^n(\bx,\bu), h_{n+1}(\bx,\bu) \right\}} &\leq \mathrm{e}^{-\beta h^n(\bx,\bu)} + \mathrm{e}^{-\beta h_{n+1}(\bx,\bu)}\nonumber\\
    &\leq 2 \mathrm{e}^{-\beta \underset{}{\min}\left\{ h^n(\bx,\bu), h_{n+1}(\bx,\bu) \right\}}
    \nonumber
\end{align}
Taking the logarithm on both sides and dividing by $-\beta$, we obtain:
\begin{align}
   & \underset{}{\min}\left\{ h^n(\bx,\bu), h_{n+1}(\bx,\bu) \right\}\nonumber\\
    &\geq -\frac{1}{\beta} \ln\left( \mathrm{e}^{-\beta h^n(\bx,\bu)} + \mathrm{e}^{-\beta h_{n+1}(\bx,\bu)} \right)\nonumber\\
    &\geq \underset{}{\min}\left\{ h^n(\bx,\bu), h_{n+1}(\bx,\bu) \right\} + \frac{\ln 2}{\beta}
\end{align}
Using the fact that $b_{n+1}\leq 1/2$, we have
\begin{align}
    -\frac{1}{\beta} \ln\left( \mathrm{e}^{-\beta h^n(\bx,\bu)} + \mathrm{e}^{-\beta h_{n+1}(\bx,\bu)} \right)\geq h^{n+1}(\bx,\bu)
\end{align}
Consequently,
\begin{align}
    h^{n+1}(\bx,\bu) \leq \underset{}{\min}\left\{ h^n(\bx,\bu), h_{n+1}(\bx,\bu) \right\}
    \label{eqn:subcase3b_inequality}
\end{align}
Thus, if $h^{n+1}(\bx,\bu) \geq 0$, then $\underset{}{\min}\left\{ h^n(\bx,\bu), h_{n+1}(\bx,\bu) \right\} \geq 0$. This implies that $(\bx,\bu) \in \mathcal{A}_s^{n+1}$. Therefore, $\mathcal{A}_s^{n+1} \supseteq \mathcal{S}_h^{n+1}$. Taking the limit as $\beta \rightarrow \infty$ on both sides of \eqref{eqn:subcase3b_inequality}, we get:
\begin{align}
    \underset{\beta\rightarrow\infty}{\lim}\;h^{n+1}(\bx,\bu) = \underset{}{\min}\left\{ h^n(\bx,\bu), h_{n+1}(\bx,\bu) \right\}
\end{align}
This shows that $\underset{\beta\rightarrow\infty}{\lim}\;\mathcal{S}_h^{n+1} = \mathcal{A}_s^{n+1}$.

By mathematical induction, we have shown that at each level, whether we have a union or an intersection, the set $\mathcal{A}_s^{n+1}$ is a superset of $\mathcal{S}_h^{n+1}$, and as $\beta \rightarrow \infty$, $\mathcal{S}_h^{n+1}$ converges to $\mathcal{A}_s^{n+1}$. Hence, the theorem is proved.

\end{proof}

}

The main motivation for using the expression $h(\bx,\bu)$ in \eqref{eqn:ubf} comes from the log-sum-exp (LSE) expressions that provide a smooth approximation of the $\max$ and $\min$ functions that characterize the union and intersection operations, respectively.

\subsection{UBF based Quadratic Programs (UBF-QP)\label{subsec:ubf_based_quadratic_programs}}
We now present a methodology to synthesize safe and stabilizing controllers via QPs.
At first glance, it appears that the expression \eqref{eqn:ubf_condition} would be a nonlinear function of $\bu$ in general, for a given $\bx$ when input constraints $h_i(\bu)$ are present. Consequently, it is not possible to synthesize inputs via QPs by modifying the inputs in the control space. To address this limitation, if the input constraints characterized by $h_i(\bu)$ are present, we modify the associated integral controller instead of modifying the control input. In particular, given a feedback controller (which is not necessarily safe or stabilizing), the modified integral controller is given by
\begin{align}
\dot{\bu}=\boldsymbol{\tau}(\bx,\bu)+\bv(\bx,\bu)
\label{eqn:int_controller}
\end{align}
where $\bv(\bx,\bu)$ is the auxiliary control input that ensures safety, stability, and input constraint satisfaction that is yet to be designed. 
\subsubsection{Computation of $\boldsymbol{\tau}(\bx,\bu)$ using Newton-Raphson Flow\label{subsec:integral_controller_newton_raphson_flow}}
We briefly discuss a method to compute the nominal integral controller using the Newton-Raphson flow presented in \cite{wardi2024tracking_integral_control_law}.
{
Consider the output equation $\by=c(\bx)$ where $c:\mathbb{R}^n\rightarrow\mathbb{R}^m$ is a continuously differentiable function. Let $T > 0$ be a fixed value. We define the predicted state trajectory $\boldsymbol{\xi}(\tau)$, where $\tau \in [t, t+T]$, by the differential equation:
\begin{align}
\dot{\boldsymbol{\xi}}(\tau) = F(\boldsymbol{\xi}(\tau), \bu(t))
\label{eqn:prelim_int_controller}
\end{align}
 with the initial condition $\boldsymbol{\xi}(t) = \bx(t)$.
Furthermore, while integrating \eqref{eqn:prelim_int_controller} we keep $t$ fixed and $\tau$ as a variable. 
Define $\tilde{\by}(t+T)$ by:
\begin{align}
\tilde{\by}(t+T) =c(\boldsymbol{\xi}(t+T))=: d(\bx(t), \bu(t))
\label{eqn:y_tilde}
\end{align}
The integral controller is given by:
\begin{align}
\dot{\bu}(t) = \eta \left(\frac{\partial d}{\partial \bu}(\bx(t), \bu(t))\right)^{-1} (\boldsymbol{r}(t+T) - d(\bx(t), \bu(t)))
\label{eqn:int_controller_raphson_flow}
\end{align}
where $\boldsymbol{r}(t)$ is the reference signal and $\eta > 0$ is the controller gain. In this case, $\boldsymbol{\tau}(\bx,\bu)$ in \eqref{eqn:int_controller} is equal to RHS of \eqref{eqn:int_controller_raphson_flow}.
Consider the case where $\eta = 1$ in \eqref{eqn:int_controller_raphson_flow}. The selection of the prediction horizon $T$ can significantly influence the tracking performance. Generally, a smaller $T$ is preferable to ensure minimal prediction errors, which may result in reduced tracking errors compared to a larger $T$ \cite{wardi2024tracking_integral_control_law}.

}
\begin{assumption}
    \normalfont We assume that the origin $\mathbf{0}_n\in\operatorname{Int}\mathcal{S}_h$ (note that $\mathcal{S}_h\subseteq\mathcal{A}_s$). In addition, we assume the existence of a known nominal integral control law $\dot{\bu}=\boldsymbol{\tau}(\bx,\bu)$ for \eqref{eqn:nonlinear_system_dynamics}. Note that methods for synthesizing such integral control laws for these systems have been developed in the literature\cite{wardi2024tracking_integral_control_law}. 
\label{assumption:integral_control_law_known_apriori_ass_5}
\end{assumption}
\subsubsection{UBF-QP}
The expression for $\dot{h}(\bx,\bu)$ can be written as the sum of three terms as follows:
\begin{align}
\dot{h}(\bx,\bu)=& {\frac{1}{\beta H_{N+1}(\bx,\bu)}\left(\frac{\partial H_{N+1}(\bx,\bu)}{\partial \bx}\right)F(\bx,\bu)}\nonumber\\
&+{\frac{1}{\beta H_{N+1}(\bx,\bu)}\frac{\partial H_{N+1}(\bx,\bu)}{\partial \bu}\left(\boldsymbol{\tau}(\bx,\bu)+\bv(\bx,\bu)\right)}\nonumber\\
&=P^a(\bx,\bu)+P^b(\bx,\bu)\bv(\bx,\bu)
\end{align}
where $P^a(\bx,\bu)$ and $P^b(\bx,\bu)$ are given by
\begin{align}
   P^a(\bx,\bu):=&\frac{1}{\beta H_{N+1}(\bx,\bu)}\left(\frac{\partial H_{N+1}(\bx,\bu)}{\partial \bx}F(\bx,\bu)\right.\nonumber\\
   &\left.+\frac{\partial H_{N+1}(\bx,\bu)}{\partial \bu}\boldsymbol{\tau}(\bx,\bu)\right) \nonumber\\
   P^b(\bx,\bu):=&\frac{1}{\beta H_{N+1}(\bx,\bu)}\frac{\partial H_{N+1}(\bx,\bu)}{\partial \bu}
\end{align}
Furthermore, $\dot{H}_N(\bx,\bu)$ can be computed recursively as follows:
\begin{align}
& \dot{H}_0(\bx,\bu)=\left\{\begin{array}{ll}
     &-\beta\mathrm{e}^{-\beta (\dot{V}(\bx)-P(\bx))}(\ddot{V}(\bx)+\dot{P}(\bx)), \; i=0  \\
     & \beta\mathrm{e}^{\beta h_i(\bx)}\dot{h}_i(\bx), \quad i \in \ix\\
     &\beta\mathrm{e}^{\beta h_i(\bu)}\dot{h}_i(\bu), \quad i \in \iu
\end{array}\right.\nonumber\\
& \dot{H}_{j+1}(\bx,\bu)=\nonumber\\
&=\left\{\begin{array}{ll} 
\dot{H}_j(\bx,\bu)+\dot{c}_1(-\dot{V}(\bx)-P(\bx)), \quad j=N+1\\
\dot{H}_j(\bx,\bu)+\dot{c}_1(h_{j+1}(\bx)), \quad j\in \pn\\
\frac{\frac{\dot{H}_j(\bx,\bu)}{\left((H_j(\bx,\bu))^{-1}\right)^2}  +\frac{\dot{c}_1(h_{j+1}(\bx))}{\left(c_1(h_{j+1}(\bx))\right)^2}}{\left((H_j(\bx,\bu))^{-1}+(c_1(h_{j+1}(\bx)))^{-1}\right)^2}, \quad j\in \qn\\
\dot{H}_j(\bx,\bu)+\dot{c}_1(h_{j+1}(\bu)), \quad j\in \pu\\
\frac{\left(\frac{\dot{H}_j(\bx,\bu)}{\left(H_j(\bx,\bu)\right)^2}+\frac{\dot{c}_1(h_{j+1}(\bu)}{\left(c_1(h_{j+1}(\bu)\right)^2}\right)}{\left((H_j(\bx,\bu))^{-1}+(c_1(h_{j+1}(\bu)))^{-1}\right)^2}, \quad j\in \qu
\end{array}\right.
\end{align}
Clearly, the integral controller term appears in the expression for $h_i(\bu)$ defined in $\dot{H}^L_N(\bx,\bu)$. 
Note that for a given $(\bx,\bu)$, the term $\dot{H}^L_N(\bx,\bu)$ is a linear function of the auxiliary control input $\bv(\bx,\bu)$. 
For a given system \eqref{eqn:nonlinear_system_dynamics}, $\boldsymbol{\tau}(\bx,\bu)$ and a UBF $h(\bx,\bu)$, denote by $p^h(\bx)$ and $q^h(\bx,\bu)$ as follows:
\begin{subequations}
    \begin{align}
    &\boldsymbol{p}^h(\bx,\bu)=\frac{\partial h(\bx,\bu)}{\partial\bu},\label{eqn:ph}\\
    &q^h(\bx,\bu)=\left(\frac{\partial h(\bx,\bu)}{\partial \bx}F(\bx,\bu)\right.\nonumber\\
   &\left.+\frac{\partial h(\bx,\bu)}{\partial \bu}\boldsymbol{\tau}(\bx,\bu)\right) +\alpha(h(\bx,\bu))\label{eqn:qh}
\end{align}
\end{subequations}
Consider the following integral control law:
\begin{align}
\dot{\bu}=\boldsymbol{\tau}(\bx,\bu)+\bv^\star(\bx,\bu)
\label{eqn:ubf_integral_controller}
\end{align}
with initial condition $\bu(0)=k(\bx)$ where $\bv^\star(\bx,\bu)$ is the minimizer of the following UBF-QP:
\begin{align}
    \textbf{UBF-QP}:\begin{cases}
&\bv^\star(\bx,\bu):=\;\;\underset{\bv\in\mathbb{R}^m}{\text{argmin}}\;\;\|\bv\|^2\\
&\left(\boldsymbol{p}^h(\bx,\bu)\right)^\mathrm{T}\bv+q^h(\bx,\bu)\geq 0\\
    \end{cases}
    \label{eqn:ubf_qp}
\end{align}
We now have the following results:
\begin{theorem}
    \normalfont The UBF-QP is feasible if  $h$ is a UBF
    \label{thm:ubf_qp_feasible}. Furthermore, $\bv^\star$ is Lipschitz continuous in $(\bx,\bu)$.
    \label{thm:ubf_qp_lipschitz_controller}
\end{theorem}
\begin{proof}
{     
Given that $ h $ is a UBF, by Definition \ref{defn:ubf}, it satisfies the condition \eqref{eqn:ubf_condition} for all $(\bx,\bu) \in \mathcal{A}$
where $ \alpha\in\mathcal{K}_\infty$. This condition translates to
\begin{align}
p^h(\bx, \bu)^\mathrm{T} \bv \geq q^h(\bx, \bu).
\label{eqn:qp_constraint}
\end{align}
Since $ h $ is a UBF, it ensures that $ \dot{h}(\bx, \bu) + \alpha(h(\bx, \bu)) \geq 0 $ can be satisfied. Therefore, there exists a $ \bv $ such that the constraint \eqref{eqn:qp_constraint} holds. This implies that the feasible set of the UBF-QP is nonempty, and thus the UBF-QP is feasible.
The UBF-QP \eqref{eqn:ubf_qp} is a convex quadratic program with a linear constraint.
 The analytical solution to this QP is:
\begin{align}
\bv^\star(\bx, \bu) = \begin{cases}
\displaystyle \frac{q^h(\bx,\bu)}{\| p^h(\bx, \bu) \|^2} p^h(\bx, \bu), & \text{if } q^h(\bx,\bu) > 0, \\
0, & \text{if } q^h(\bx,\bu) \leq 0.
\end{cases}
\label{eqn:optimal_control}
\end{align}
We need to show that $ \bv^\star(\bx, \bu) $ is Lipschitz continuous in $ (\bx, \bu) $.
First, note that $ \boldsymbol{p}^h(\bx,\bu) $, $ q^h(\bx,\bu) $, and $ h(\bx, \bu) $ are continuously differentiable functions of $ (\bx, \bu) $, given the smoothness of $ h $ and the system dynamics \eqref{eqn:nonlinear_system_dynamics}. 
We consider two cases:

\textbf{Case 1}: $ q^h(\bx, \bu) > 0 $:
In this case, the optimal control that solves \eqref{eqn:ubf_qp} is given by:
\begin{align}
\bv^\star(\bx, \bu) = \frac{q^h(\bx,\bu)}{\| \boldsymbol{p}^h(\bx,\bu) \|^2} \boldsymbol{p}^h(\bx,\bu).
\end{align}
Both $ q^h(\bx,\bu) $ and $ \boldsymbol{p}^h(\bx,\bu) $ are Lipschitz continuous, and $ \| \boldsymbol{p}^h(\bx,\bu) \|^2 $ is Lipschitz continuous and bounded away from zero (since $ \boldsymbol{p}^h(\bx,\bu) $ is nonzero when $ q^h(\bx,\bu) > 0 $). Therefore, the quotient $ \frac{q^h(\bx,\bu)}{\| \boldsymbol{p}^h(\bx,\bu) \|^2} $ is Lipschitz continuous. The product of Lipschitz continuous functions is Lipschitz continuous, so $ \bv^\star(\bx, \bu) $ is Lipschitz continuous in this case.

\textbf{Case 2}: $ q^h(\bx,\bu) \leq 0 $:
Here, $ \bv^\star(\bx, \bu) = 0 $, which is trivially Lipschitz continuous.
At the boundary where $ q^h(\bx,\bu) = 0 $, we need to ensure that $ \bv^\star(\bx, \bu) $ does not have a discontinuity. As $ q^h(\bx,\bu) \rightarrow 0^+ $, $ \bv^\star(\bx, \bu) \rightarrow 0 $ because $ q^h(\bx,\bu) $ tends to zero and $ \boldsymbol{p}^h(\bx,\bu) $ is bounded.
Therefore, $ \bv^\star(\bx, \bu) $ is continuous at $ q^h(\bx,\bu) = 0 $.

}

\end{proof}
\begin{remark}
    \normalfont From Theorem \ref{thm:ubf_qp_lipschitz_controller}, $\bv^\star(\bx,\bu)$ (obtained from \eqref{eqn:ubf_qp}) is Lipschitz continuous. Furthermore, as $\boldsymbol{\tau}(\bx,\bu)$ is a continuously differentiable function, it follows from Picard-Lindel\"{o}f theorem \cite{hartman2002ordinary}, the existence and uniqueness of the closed loop solution of \eqref{eqn:ubf_integral_controller} is guaranteed around a neighborhood of the current state $\bx$.
    \end{remark}
\begin{remark}
    \normalfont Note that the proposed UBF-QP is not devoid of the difficulty of finding a UBF (Definition \ref{defn:ubf}). Rather, it provides a framework to address the limitations presented in Section \ref{sec:motivation} which are common in practical applications. 
\end{remark}
\begin{remark}
    \normalfont Note that the UBF-QP \eqref{eqn:ubf_qp} can be easily modified to handle time-varying safety specifications. This can be particularly important for some practical applications. For instance, consider a collision avoidance problem where the task is to avoid static/dynamic obstacles and reach the goal position. In that case, it would be advantageous to consider the obstacles present only in the field of view. Furthermore, in case of an actuator failure, it would be restrictive to have both stability and safety guarantees instead of safety only.
\end{remark}
The following theorem provides a method for synthesizing safe and stabilizing feedback controllers via a UBF.
\begin{theorem}
\normalfont Under Assumptions \ref{assumption:stabilizing_assum_1}-\ref{assumption:integral_control_law_known_apriori_ass_5}, for the system \eqref{eqn:nonlinear_system_dynamics}, if there exists a UBF $h(\bx,\bu)$, then any Lipschitz continuous controller $\bu\in \mathcal{K}_{\text{UBF}}(\bx)$ renders the set $\mathcal{S}_h\subset\mathcal{A}_s$ (where $\mathcal{A}_s$ is defined in \eqref{eqn:state_and_input_constraint_set_complex}) characterized by the UBF; $\mathcal{S}_h=\{(\bx,\bu)\;:\;h(\bx,\bu)\geq 0\}$ forward invariant.
\end{theorem}
\begin{proof}
    \normalfont 
 The set of interest is $\mathcal{S}_h = \{ (\bx,\bu) \in \mathbb{R}^n\times\mathbb{R}^m \mid h(\bx, \bu) \geq 0,\;\;\bu\in\mathcal{K}_{\text{UBF}}(\bx)\}$. In the subsequent discussion, let $u\in\mathcal{K}_{\text{UBF}}(\bx)$.
 Under the control law $\bu \in \mathcal{K}_{\text{UBF}}(\bx)$, the closed-loop system dynamics becomes
\begin{align}
    \dot{\bx} = F(\bx, \bu),\quad\bu\in\mathcal{K}_{\text{UBF}}(\bx)
    \label{eq:closed_loop_dynamics},\quad \bx(0)=\bx_0
\end{align}
Since $F$ and $\bu$ are continuously differentiable and Lipschitz continuous (Theorem \ref{thm:ubf_qp_lipschitz_controller}), respectively, the closed loop solution of \eqref{eq:closed_loop_dynamics} exists and is unique around the neighborhood of the current state $\bx$. 
Now, consider the scalar differential equation
\begin{align}
    \dot{y}(t) = -\alpha(y(t)),\quad y(0) = h(\bx_0, \bu_0)
    \label{eq:y_dynamics}
\end{align}
where $\alpha$ is a class $\mathcal{K}_\infty$ function. From Lemma 4.4 of \cite{khalil1996nonlinear}, solution of \eqref{eq:y_dynamics} can be expressed as $y(t)=\sigma(\bx_0,t)$ where $\sigma$ is a class $\mathcal{KL}$ function. This implies that
$y(t)$ remains non-negative for all $t \geq 0$, provided $y(0) \geq 0$. Comparing with \eqref{eq:y_dynamics}, we observe that the function $h(\bx, \bu)$ satisfies $ \dot{h}(\bx, \bu) \geq \dot{y}(t)$
with $h(\bx_0, \bu_0) = y(0)$. By the Comparison Lemma \cite{khalil1996nonlinear}, if a continuous function (${h}(\bx(t),\bu(t))$ in this case) satisfies an inequality of the form $\dot{h}(\bx(t),\bu(t)) \geq \dot{y}(t)$ with $y(0)=h(\bx_0,\bu_0)$, then
\begin{align}
    h(\bx, \bu) \geq y(t)\geq 0, \quad \forall t \geq 0.
    \label{eq:h_geq_y}
\end{align}
Therefore, if the system \eqref{eqn:nonlinear_system_dynamics} starts in the set $\mathcal{S}_h$ (i.e., $h(\bx_0, \bu_0) \geq 0$), it will remain in $\mathcal{S}_h$ for $t\geq 0$. 
Hence, the result follows.
\end{proof}
\begin{corollary}
    \normalfont Consider an integral control law $\dot{\bu}=\boldsymbol{\tau}(\bx,\bu)$. Under the assumptions \ref{assumption:stabilizing_assum_1}-\ref{assumption:integral_control_law_known_apriori_ass_5}, the system \eqref{eqn:nonlinear_system_dynamics}, the corresponding integral controller $\dot{\bu}=\boldsymbol{\tau}(\bx,\bu)$ and an UBF $h(\bx,\bu)$, then the integral controller \eqref{eqn:ubf_integral_controller} where $\bv^\star$ is the solution to \eqref{eqn:ubf_qp}, renders the set $\mathcal{S}_h=\{(\bx,\bu)\in\mathbb{R}^n\times \mathbb{R}^m\;\mid\;h(\bx,\bu)\geq 0\}$ forward invariant.
\end{corollary}
\begin{proof}
    \normalfont Given the integral control law $\dot{\bu}=\boldsymbol{\tau}(\bx,\bu)+\bv(\bx,\bu)$, the condition in \eqref{eqn:ubf_condition} translates to $\left(\boldsymbol{p}^h(\bx,\bu)\right)^\mathrm{T}\bv+q^h(\bx,\bu)\geq 0$ and hence the result follows.
\end{proof}
\begin{remark}
    \normalfont The necessary and sufficient condition for the existence of the solution to \eqref{eqn:ubf_qp} is that if $\boldsymbol{p}^h(\bx,\bu)=\mathbf{0}_m$ if and only if $q^h(\bx,\bu)\geq 0$.
\end{remark}
\subsection{Existence of UBF\label{subsec:existence_of_ubf}}
In this section, we provide sufficient conditions under which there exists a UBF. This result is crucial because it assures that, under mild conditions, one can systematically construct a Lipschitz continuous feedback controller that achieves safe and stable behavior for a general class of nonlinear systems.
\begin{theorem}
    \normalfont $(\textbf{Existence of UBF})$
Under Assumptions \ref{assumption:stabilizing_assum_1}-\ref{assumption:integral_control_law_known_apriori_ass_5}, consider the system \eqref{eqn:nonlinear_system_dynamics} and let $\mathcal{A}_s$ represent a given specification for which there exists a function $h: \mathbb{R}^n\times\mathbb{R}^m \rightarrow \mathbb{R}$ that is continuously differentiable and satisfies the condition. Assume that $\mathcal{S}_h$ $\subset \mathcal{A}_s$ is a safe set, such that there exists a compact set $\mathcal{U} \subset \mathbb{R}^m$ and a locally Lipschitz controller $\bu: \mathbb{R}^m \rightarrow \mathcal{U}$ that ensures safety, i.e., forward invariance. Then, there exists a UBF. 
    \label{thm:existence_ubf}
\end{theorem}
\begin{proof}
    \normalfont {Consider the augmented nonlinear system formed by augmenting \eqref{eqn:nonlinear_system_dynamics} and the integral controller $\dot{\bu}=\boldsymbol{\tau}(\bx,\bu)$ (with $\bu(0)=\bu_0$) as follows:
    \begin{align}
        \dot{\bz}=F^a(\bz),\quad \bz(0)=[\bx_0,\;\bu_0]^\mathrm{T}
        \label{eqn:augmented_system}
    \end{align}
    where $\bz=[\bx,\;\bu]^\mathrm{T}$, $F^a(\bz)=[F(\bx,\bu),\;\boldsymbol{\tau}(\bx,\bu)]^\mathrm{T}$ and $\operatorname{Int}\mathcal{S}_h$. We define a cost function $V: \mathbb{R}^n \times\mathbb{R}^m\times \mathbb{R}_{>0} \rightarrow \mathbb{R}$ as:
\begin{align}
V(\bz, t)=\min _{s \in[0, t]} h(\varphi(s;\bz)),
\label{eqn:v_construct_based_on_h}
\end{align}
where $h$ is a candidate UBF and $\bu$ obtained from solving $\dot{\bu}=\boldsymbol{\tau}(\bx,\bu)$ (with $\bu(0)=\bu_0$) is assumed to be a safe controller (Assumption \ref{assumption:stabilizing_assum_1}). Note that $V(\bz, t)$ in \eqref{eqn:v_construct_based_on_h} represents the minimum value of $h$ along the system trajectories $\varphi(\cdot)$ \eqref{eqn:augmented_system}, given initial condition $\bz$, and final time $t\geq 0$. Note that, in \eqref{eqn:v_construct_based_on_h}, we omit the maximization over all possible controllers $\bu\in\mathbb{R}^m$ and solely utilize $\bu$ (safe controller). By extending $V$ for infinite time as $V_{\infty}(\bz):=\underset{t\rightarrow\infty}{\lim}\;V(\bz,t)$, we obtain a time-invariant function. Note that as $\bu$ is a safe controller (i.e., the compact set $\mathcal{S}_h$ remains forward invariant) and $h$ is a continuously differentiable function, $V_\infty$ exists. The zero-superlevel set of $V_{\infty}(\bz)$ constitutes the largest forward invariant set of $\dot{\bz}=F^a\left(\bz\right)$ contained within $\mathcal{S}_h=\left\{(\bx,\bu) \in \mathbb{R}^n\times\mathbb{R}^m : V_{\infty}(\bz) \geq 0\right\}$. Furthermore, since $\bu$ maintains all trajectories of \eqref{eqn:augmented_system} within $\operatorname{Int}(\mathcal{S}_h)$ for all times, and these trajectories do not approach $\partial \mathcal{S}_h$ arbitrarily closely, we have $V_{\infty}(\bx)>0$ for all $\bz \in \operatorname{Int}(\mathcal{S}_h)$. For all points in $\mathcal{S}_h$ where the gradient of $V_{\infty}$ exists:
\begin{align}
\nabla V_{\infty}(\bz)^\mathrm{T} F^a\left(\bz\right) \geq-\alpha\left(V_{\infty}(\bz)\right)
\label{eqn:v_infty_condition}
\end{align}
for any class $\mathcal{K}_{\infty}$ function $\alpha$. However, $V_{\infty}$ might not be differentiable at all points, potentially disqualifying it as a valid UBF. Nevertheless, given that $F^a, \bu$, and $h$ are locally Lipschitz, $V_{\infty}$ is differentiable almost everywhere. To obtain a valid UBF, we smoothen $V_{\infty}$. We begin by demonstrating that $V_{\infty}$ can be smoothened at the interior of $\mathcal{S}_h$, ensuring \eqref{eqn:v_infty_condition} holds for all $\bz \in \operatorname{Int}(\mathcal{S}_h)$ for the smoothed version of $V_{\infty}$.
In addition, note that there exists a smooth function $\Psi: \operatorname{Int}(\mathcal{S}_h) \rightarrow \mathbb{R}$ such that for all $\bz \in \operatorname{Int}(\mathcal{S}_h)$ \cite{lin1996smooth}:
\begin{align}
\left|V_{\infty}(\bz)-\Psi(\bz)\right|&<\min \left\{\frac{1}{2} V_{\infty}(\bz), 1\right\} \\
\nabla \Psi(\bz)^\mathrm{T} F^a\left(\bz\right) &\geq-2 \alpha\left(V_{\infty}(\bz)\right)
\end{align}
Given that $V_{\infty}(\bz)>0$ for all $\bz \in \operatorname{Int}(\mathcal{S}_h)$, it follows that $\Psi(\bz)>V_{\infty}(\bz)- \frac{1}{2} V_{\infty}(\bz)=\frac{1}{2} V_{\infty}(\bz)>0$ for all $\bz \in \operatorname{Int}(\mathcal{S}_h)$. We then extend $\Psi$ to $\partial \mathcal{S}_h$ such that $\Psi(\bz)=0$ for all $\bz \in \partial \mathcal{S}_h$. Consequently, $\Psi$ is smooth in $\operatorname{Int}(\mathcal{S}_h)$ and continuous in $\mathcal{S}_h$. Moreover, since $\alpha$ is increasing, $2 \alpha\left(V_{\infty}(\bz)\right) \leq 2 \alpha(2 \Psi(\bz))$. By defining $\bar{\alpha}(r)=2 \alpha(2 r)$, we ensure that $\bar{\alpha}$ is smooth, extended class $\mathcal{K}_{\infty}$, and for all $\bz \in \operatorname{Int}(\mathcal{S}_h)$, it holds that $\nabla \Psi(\bz)^\mathrm{T} F^a\left(\bz\right) \geq -\bar{\alpha}(\Psi(\bz))$. 
To ensure that the function $\Psi$ is well-defined on the closure of $\mathcal{S}_h$, we extend $\Psi$ to the boundary $\partial \mathcal{S}_h$ by defining
\begin{align}
    \Psi(\bz) = 0, \quad \forall \bz \in \partial \mathcal{S}_h.
    \label{eq:Psi_extension}
\end{align}
This extension guarantees that $\Psi$ remains continuous on $\mathcal{S}_h$ and retains its smoothness on $\operatorname{Int}(\mathcal{S}_h)$. With $\Psi$ now defined on $\mathcal{S}_h$, we consider $\Psi$ as our candidate UBF function, denoted by $h(\bx,\bu) = \Psi(\bz)$. From the properties established, $h$ satisfies $h(\bx,\bu) = \Psi(\bz) > 0$ for all $\bz \in \operatorname{Int}(\mathcal{S}_h)$
   since $\Psi(\bz) > \frac{1}{2} V_{\infty}(\bz) > 0$ in the interior of $\mathcal{S}_h$. In addition, $h(\bx,\bu) = \Psi(\bz) = 0$ for all $\bz \in \partial \mathcal{S}_h$
   by the extension in \eqref{eq:Psi_extension}. Consequently, we have
   \begin{align}
       \nabla h(\bx,\bu)^\mathrm{T} F^a\left(\bz\right) \geq -\bar{\alpha}(h(\bx,\bu)), \;\; \forall\;\; \bz \in \operatorname{Int}(\mathcal{S}_h).
       \label{eq:h_ubf_condition}
   \end{align}
Thus, $h$ satisfies the UBF condition within the interior of $\mathcal{S}_h$. Moreover, since $h$ is continuous on $\mathcal{S}_h$ and differentiable on $\operatorname{Int}(\mathcal{S}_h)$, we can conclude that $h$ is a valid UBF for the system \eqref{eqn:augmented_system} under the control law $\dot{\bu}=\boldsymbol{\tau}(\bx,\bu)$.

}
\end{proof}
\begin{remark}
    \normalfont Note that the sufficient condition presented in Theorem \ref{thm:existence_ubf} is reasonable in the sense that, the safe control design problem would be well posed if there exists a safe controller. This is because synthesizing safe and stabilizing controllers via UBF makes sense only if there exists a controller that ensures the forward invariance of the set defined by the specification $\mathcal{A}_s$ \eqref{eqn:specification_with_stability}. 
\end{remark}
\begin{figure}
\begin{tikzpicture}
    \input{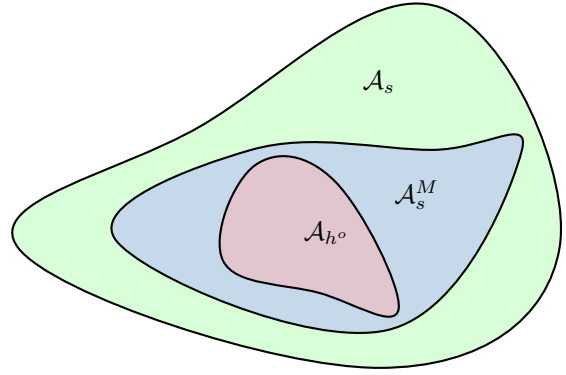}
\end{tikzpicture}
\caption{The depiction of sets $\mathcal{A}_s$, $\mathcal{A}^M_s$ and $\mathcal{A}_{h^o}$ discussed in Section \ref{subsec:high_order_ubf}}
\label{fig:sets_high_order}
\end{figure}
\subsection{UBF for systems with higher relative degrees\label{subsec:high_order_ubf}}
This section is motivated from the two current limitations of using UBF-QP. The first limitation, is that at least one of the individual barrier functions $h_i$ ($i\in\mathcal{I}_N$), may have a relative degree greater than one. Second, is that, there might not exist a control input such that the UBF condition \eqref{eqn:ubf_condition} is satisfied. In other words, $\boldsymbol{p}^h(\bx,\bu)=0$ (in \eqref{eqn:ph}) does not imply that $q^h(\bx,\bu)\geq 0$ (in \eqref{eqn:qh}) for $(\bx,\bu)\in\mathcal{S}_h$. To address these two limitations, in this section, we present the notion of HO-UBF as follows:

    
Let the function $h_i(\bx)$ that characterizes the safe set $\mathcal{S}_i=\{\bx\in\mathbb{R}^n\,|\;h_i(\bx)\geq 0\}$ be of relative degree $m_i(\geq 1)$ where $i\in\mathcal{P}_\bx\cup\mathcal{Q}_\bx$. If $h_i(\bx)$ are $m_i$ times continuously differentiable, we define a sequence of functions  as follows:
\begin{subequations}
    \begin{align}
&\Phi^j_i(\bx)=\dot{\psi}^{j-1}_i(\bx)+\alpha_i(\psi^{j-1}_i(\bx)),\quad \Phi^1_i(\bx)=h_i(\bx)\\
&\psi^j_i(\bx)\geq \Phi^j_i(\bx),\quad j\in \mathcal{I}_{(1,m_i)},\quad i\in\mathcal{P}_\bx\cup\mathcal{Q}_\bx\label{eqn:less_conservative_condition_houbf}
\end{align}
\end{subequations}
where $\alpha_i$ is a class $\mathcal{K}_\infty$ function. Furthermore, we define the set $\mathcal{S}^{m_i}_i$ as follows:
\begin{align}
\mathcal{S}^{j}_i=\begin{cases}
    \{\bx\in\mathbb{R}^n|\;\Phi^{j}_i(\bx)\geq 0\},\quad i\in\mathcal{P}_\bx\cup\mathcal{Q}_\bx,\quad j\in\mathcal{I}_{(1,m_i)}\\
    \{\bu\in\mathbb{R}^m|\;h_i(\bu)\geq 0\},\quad i\in\mathcal{P}_\bu\cup\mathcal{Q}_\bu
    \end{cases}
\label{eqn:sets_high_order_ubf}
\end{align}
Consequently, the set $\mathcal{S}_i$ is defined by:
\begin{align}
    \mathcal{S}_i=\begin{cases}
   \cap_{j=1}^{m_i} \mathcal{S}_i^j,\quad\quad\quad\quad\quad\quad\quad\; i\in\mathcal{P}_\bx\cup\mathcal{Q}_\bx,\;\; j\in\mathcal{I}_{(1,m_i)}\\
    \{\bu\in\mathbb{R}^m|\;h_i(\bu)\geq 0\},\quad i\in\mathcal{P}_\bu\cup\mathcal{Q}_\bu
\end{cases}\label{eqn:modified_high_set}
\end{align}
Finally, the modified set $\mathcal{A}^M_{s}$ is defined by
\begin{align}
   \mathcal{A}^M_s=(\mathcal{S}\cap\mathcal{U})\cap\mathcal{S}_V,
\label{eqn:modified_state_and_input_constraint_set_complex}
\end{align}
where the sets $\mathcal{S}$ and $\mathcal{U}$ are given by
\begin{align}
    &\mathcal{S}=\oplus_{i=1}^{N_x} \mathcal{S}_i,\quad\mathcal{U}=\oplus_{i=1}^{N_u} \{\bu\in\mathbb{R}^m|\;h_i(\bu)\geq 0\}\nonumber
\end{align}
where $\mathcal{S}_i$ is defined in \eqref{eqn:modified_high_set}.
 If at least one of $m_i>1$ for any $i\in\mathcal{P}_\bx\cup\mathcal{Q}_\bx$, then the UBF is defined as in \eqref{eqn:ubf} with the only difference being that $h_i(\bx)$ is replaced by $\Phi^{m_i}_i(\bx)$ and is given by 
 \begin{align}
&h^o(\bx,\bu)=c_2\left(H_{N+1}(\bx,\bu)  \right)+c_2(b_N)
\label{eqn:ubf_equation_for_high_order}
    \end{align}
    where $b_N=\left({\prod_{i=1}^{|\mathcal{P}'_N|}\;|\mathcal{P}'_N(i)|}\right)^{-1}$, and $H_{N+1}(\bx,\bu)$ is recursively defined by
    \begin{subequations}
    \begin{align}
    & H_0(\bx,\bu)=\left\{\begin{array}{ll}
     &c_1(-\beta( \dot{V}(\bx)-P(\bx))), \quad\quad\; \text{if } i=N+1  \\
     & c_1(\beta \Phi^{m_i}_i(\bx)), \quad \quad\quad\quad\quad\quad\;\text{if } i\in \ix\\
     &c_1(\beta h_i(\bu)), \quad\quad\quad\quad\quad\quad\;\;\;\; \text{if } i\in \iu
\end{array}\right.\\
&H_{j+1}(\bx,\bu)\nonumber\\
&=\left\{\begin{array}{ll} 
H_j(\bx,\bu)+c_1(\dot{V}(\bx)), \quad\quad\quad\quad\quad\;\;\;\;\;\;\; j=N+1\\
H_j(\bx,\bu)+c_1(\Phi^{m_{j+1}}_{j+1}(\bx)), \quad\quad\quad\quad\;\;\;\; j\in \pn\\
\left((H_j(\bx,\bu))^{-1}+(c_1(\Phi^{m_{j+1}}_{j+1}(\bx)))^{-1}\right)^{-1}, \quad\; j\in \qn\\
H_j(\bx,\bu)+c_1(h_{j+1}(\bu)), \quad\quad\quad\quad\quad\;\; j\in \pu\\
\left((H_j(\bx,\bu))^{-1}+(c_1(h_{j+1}(\bu)))^{-1}\right)^{-1}, \quad\quad j\in \qu\nonumber
\end{array}\right.
\end{align}
\label{eqn:ubf_math_equation_ho}
    \end{subequations}
    where $\beta>0$. Finally, the set $\mathcal{A}_{h^o}$ is defined by
\begin{align}
\mathcal{A}_{h^o}:=\{(\bx,\bu)\in\mathbb{R}^n\times\mathbb{R}^m\;\mid h^o(\bx,\bu)\geq 0\}
\end{align}
Using \eqref{eqn:modified_high_set}, \eqref{eqn:modified_state_and_input_constraint_set_complex} and Theorem \ref{thm:subset}, it can be shown that $\mathcal{A}_{h^o}\subseteq\mathcal{A}^M_{s}\subseteq\mathcal{A}_s$.


{
\begin{definition}
    \normalfont \textbf{(High Order UBF)} Let $h^o$ be defined as in \eqref{eqn:ubf_equation_for_high_order}. Consider the following two scenarios. First, $p^{h^0}(\bx,\bu)=0$ if and only if $q^{h^o}(\bx,\bu)\geq 0$. Then $\Pi^0(\bx,\bu)=h^o(\bx,\bu)$ is a High Order UBF (HO-UBF), if there exists a class $\mathcal{K}_\infty$ function such that 
    \begin{align}
       \dot{h}^o(\bx,\bu)\geq -\alpha\left(h^o(\bx,\bu)\right) 
    \end{align}
    Second, if this condition does not hold true (i.e., $p^{h^0}(\bx,\bu)=0 \iff q^{h^o}(\bx,\bu)\geq 0$ ), we define the sequence of functions $\Pi^i$ as follows 
    \begin{align}
       & \Pi^0(\bx,\bu)={h}^o(\bx,\bu),\nonumber\\
        & \Pi^i(\bx,\bu)=\dot{\Pi}^{i-1}(\bx,\bu)+\alpha^{i-1}({\Pi}^{i-1}(\bx,\bu)),\;\;\forall\;\; i\in\mathcal{I}_{(1,m)}
    \end{align}
    where $\alpha^i$ (for $i\in\mathcal{I}_{m-1}$) are class $\mathcal{K}_\infty$ functions and $m>1$. If $\Pi^m(\bx,\bu)$ is such that the following holds true
    \begin{align}
    p^{\Pi^m}(\bx,\bu)=0\;\Leftrightarrow\;q^{\Pi^m}(\bx,\bu)\geq 0
    \end{align}
    where $p^{\Pi^m}(\bx,\bu)$ and $q^{\Pi^m}(\bx,\bu)$ are given by
    \begin{align}
    &p^{\Pi^m}(\bx,\bu)=\frac{\partial \Pi^m(\bx,\bu)}{\partial\bu},\nonumber\\
        &q^{\Pi^m}(\bx,\bu)=\frac{\partial \Pi^m}{\partial\bx}F(\bx,\bu)+\frac{\partial \Pi^m}{\partial\bu}\boldsymbol{\tau}(\bx,\bu)\nonumber
    \end{align}
    Then, $\Pi^m(\bx,\bu)$ is a High Order UBF if there exists a class $\mathcal{K}_\infty$ function such that 
    \begin{align}
       \dot{\Pi}^m(\bx,\bu)\geq -\alpha\left(\Pi^m(\bx,\bu)\right) 
       \label{eqn:ho_ubf_condition}
    \end{align}
    \label{defn:ho_ubf}
\end{definition}
We define the set $\mathcal{A}_{\Pi^m}$ as follows:
\begin{align}
\mathcal{A}_{\Pi^m}=\cap_{i=0}^m\;\mathcal{R}^i_{\Pi^m}
    \label{eqn:intersect_ho_ubf}
\end{align}
where $\mathcal{R}^i_{\Pi^m}=\left\{(\bx,\bu)\in\mathbb{R}^n\times\mathbb{R}^m\;\mid \Pi^i(\bx,\bu)\geq 0\right\}$. 
\begin{assumption}
    \normalfont We assume that there exists a $m>1$ such that $\Pi^m(\bx,\bu)$ is a HO-UBF.
\end{assumption}

}
Subsequently, we define the set $K_{\text{HO-UBF}}$ as follows:
\begin{align}
    K_{\text{HO-UBF}}=\left\{\bu\in\mathbb{R}^m\;| \dot{\Pi}^m(\bx,\bu)\geq -\alpha^m(\Pi^m(\bx,\bu))\right\}
\end{align}
\begin{remark}
\normalfont For given barrier functions $h_i(\bx)$ ($i\in\mathcal{I}_N$), the notion of HO-UBF is more general than that proposed in \cite{xiao2021high_order_cbf_1,tan2021high_order_2} (Definition \ref{def:high_order_cbf}). Particularly, HO-UBF translates to the HO-CBF when the condition \eqref{eqn:less_conservative_condition_houbf} changes to $\psi^j_i(\bx)= \Phi^j_i(\bx) $ which is more conservative. This is illustrated by considering the following example
\end{remark}
\begin{example}
    \normalfont Consider the nonlinear system given by $\dot{x}_1 = -x_2^2 + 4$ and $\dot{x}_2 = x_1 + u$. Define the safe set as $\mathcal{S} = \{\bx=(x_1,\; x_2)\; | x_1 \geq 0\}$. The High Order CBF (HO-CBF) based on \cite{xiao2021high_order_cbf_1} (Definition \ref{def:high_order_cbf}) $\psi^2(\bx)$ is given by
\begin{align}
    \psi^1(\bx) = 1 - \mathrm{e}^{-x_1}, \quad \psi^2(\bx) = \dot{x}_1 + \alpha_1(\psi^1(\bx))
\end{align}
where $\alpha_1\in\mathcal{K}_\infty$. Since the relative degree $m_1=2$ and $\mathcal{P}_\bx=\{1\}$ (as the safety specification consists of only one set), the HO-UBF $\Phi^2_i(\bx) \geq \psi^i(\bx)$, is defined as
\begin{align}
 \Phi^1_1(\bx) = 1 - \mathrm{e}^{-x_1}, \;\;  \Phi^2_1(\bx) = \dot{\psi}^1_1(\bx) + \alpha_1(\psi^1_1(\bx)), \nonumber
\end{align}
where $\psi^1_1(\bx) \geq \Phi^1_1(\bx)$. Choosing $\psi^1_1(\bx) = 1 + \mathrm{e}^{-x_1} \geq \Phi^1_1(\bx)$ leads to $\Phi^2_1(\bx) = -\mathrm{e}^{-x_1}(-x_2^2 + 4)$. Consequently, the sets $\mathcal{S}^\psi$ and $\mathcal{S}^\Phi$ are defined as
\begin{align}
&\mathcal{S}^\psi = \{\bx | \psi^1(\bx) \geq 0 \text{ and } \psi^2(\bx) \geq 0\}\nonumber\\
&\quad\;\; =\{\bx | x_1\geq 0,\;x_2\in\mathbb{R}\}\cap \{\bx | x_1\in\mathbb{R},\;x_2\in[-2,2]\}\nonumber\\
&\quad \;\;= \{\bx | x_1\geq 0,\;\; x_2 \in [-2, 2]\}\subset \mathcal{S}\nonumber\\
&\mathcal{S}^\Phi = \{\bx | \Phi^1_1(\bx) \geq 0 \text{ and } \Phi^2_1(\bx) \geq 0\}\nonumber\\
&\quad\;\; =\{\bx | x_1\geq 0,\;x_2\in\mathbb{R}\}\cap\nonumber\\
&\quad\quad\;\;\;\{\bx | x_1\in\mathbb{R},\;x_2\in(\infty, -2] \cup [2, \infty)\}\nonumber\\
&\quad\;\;= \{\bx |x_1\geq 0,\; x_2 \in (\infty, -2] \cup [2, \infty)\}\subset\mathcal{S}\nonumber
\end{align}
\end{example}
Clearly, the safe set characterized by HO-UBF $\mathcal{S}^\Phi$ (i.e., $\{\bx\;|\;\mathcal{S}^\Phi\geq 0\}$) is larger than the set characterized by HO-CBF $\mathcal{S}^\psi$ (i.e., $\{\bx\;|\;\mathcal{S}^\psi\geq 0\}$) which highlights the conservative nature of HO-CBF \cite{xiao2021high_order_cbf_1,tan2021high_order_2}.

{\color{black}
\begin{theorem}
\normalfont Consider a High Order UBF $\Pi^m$ (Definition \eqref{defn:ho_ubf}) for $m\geq 0$. If $(\bx_0,\bu_0)\in\mathcal{A}_{\Pi^m}$, then any Lipschitz continuous controller $\bu\in K_{\Pi^m}$ ensures that $(\bx,\;\bu)\in\mathcal{A}_{\Pi^m}$ for $t\geq 0$.
\end{theorem}
\begin{proof}
\normalfont Consider the first scenario where $p^{h^0}(\bx,\bu)=0$ if and only if $q^{h^o}(\bx,\bu)\geq 0$ for $h^o$ defined in \eqref{eqn:ubf_equation_for_high_order}. In that case $\Pi^0=h^o$ is a High Order UBF. Assume $(\bx_0,\bu_0)\in\mathcal{A}_{\Pi^m}$. Consequently, if there exists a Lipschitz controller $\bu\in K_{\text{HO-UBF}}(\bx,\bu)$, then the set $\mathcal{A}_{\Pi^m}$ is forward invariant. 

Now consider the second scenario where $\Pi^m$ is a High Order UBF and assume $(\bx_0,\bu_0)\in\mathcal{R}^m_{\Pi^m}$. The condition \eqref{eqn:ho_ubf_condition} implies that $\Pi^{m}(\bx,\bu)\geq 0$ for $(\bx,\bu)\in\mathcal{R}^m_{\Pi^m}$. Consequently,
\begin{align}
\dot{\Pi}^{m-1}_i(\bx,\bu)+\alpha^{m-1}\left(\Pi^{m-1}(\bx,\bu)\right) \geq 0,\quad i\in\mathcal{P}_\bx\cup\mathcal{Q}_\bx \nonumber
\end{align}
Given that $(\bx_0,\bu_0)\in \mathcal{A}_{\Pi^m}$, we know $\Pi^{m-1}\left(\bx_0,\bu_0\right)\geq 0$. Since $\Pi^{m-1}(\bx,\bu)$ represents $\Pi^{m-1}(t)$ explicitly, it follows that $\Pi^{m-1}(\bx,\bu)\geq 0$ for $t\geq 0$, and therefore $(\bx,\bu)\in\mathcal{R}^{m-1}_{\Pi^m}$. Following a similar iteration process,
we can demonstrate that $(\bx,\bu)\in\mathcal{R}^i_{\Pi^m}$ and for $t\geq 0$ and $i\in\mathcal{I}_m$. Thus, the sets $\mathcal{R}^0_{\Pi^m},\;\mathcal{R}^1_{\Pi^m},\ldots,\mathcal{R}^m_{\Pi^m}$ are individually forward invariant. As a result, their intersection $\mathcal{A}_{\Pi^m}$ \eqref{eqn:intersect_ho_ubf} is also forward invariant.
\end{proof}

}
\section{Numerical simulations\label{sec:numerical_simulations}}
In this section, we present numerical simulations to validate the effectiveness of the proposed Universal Barrier Function (UBF) and High Order Universal Barrier Function (HO-UBF) in ensuring safety and stability for nonlinear control systems subject to complex state and input constraints. We consider three examples: a single integrator system, a double integrator, and a quadrotor system. 

\subsection{Single Integrator System with UBF}

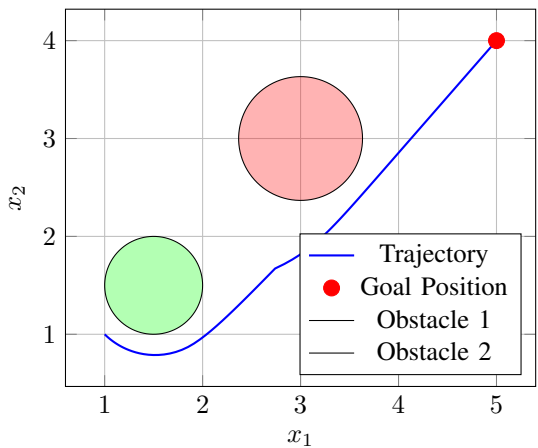
\begin{figure}
    \centering
\begin{tikzpicture}
    \begin{axis}[
    scale=0.6,
        xlabel={$x_1$},
        ylabel={$x_2$},
        title={},
        legend pos= south east,
        grid=both,
        axis equal image,
        width=12cm,
        height=12cm,
    ]

    \addplot[
        blue,
        thick,
    ] table [x=x1, y=x2, col sep=comma] {csv/trajectory.csv};
    \addlegendentry{Trajectory}

    \addplot[
        only marks,
        mark=*, 
        mark size=3pt,
        red,
    ] coordinates {(5,4)};
    \addlegendentry{Goal Position}

    \addplot[
        red,
        fill=red, 
        fill opacity=0.3,
        draw=black,
        domain=0:360,
        samples=100,
        parametric,
    ] 
    (
        {3 + 0.632455532*cos(x)}, 
        {3 + 0.632455532*sin(x)}
    );
    \addlegendentry{Obstacle 1}

    \addplot[
        green,
        fill=green, 
        fill opacity=0.3,
        draw=black,
        domain=0:360,
        samples=100,
        parametric,
    ] 
    (
        {1.5 + 0.5*cos(x)}, 
        {1.5 + 0.5*sin(x)}
    );
    \addlegendentry{Obstacle 2}

    \end{axis}
\end{tikzpicture}    \caption{Trajectory of single integrator system}
\label{fig:traj_single_integrator}
\end{figure}
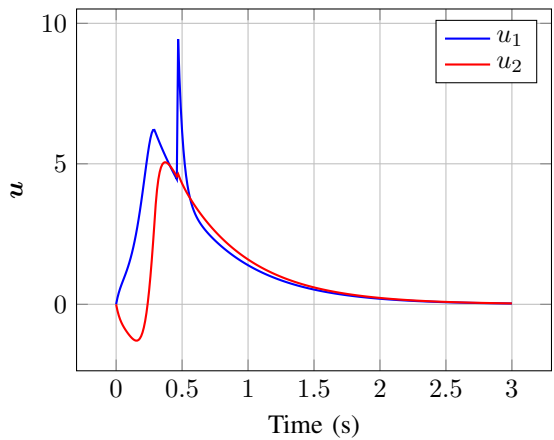
\begin{figure}
    \centering
\begin{tikzpicture}
    \begin{axis}[
        scale=0.75,
        xlabel={Time (s)},
        ylabel={$\bu$},
        grid=both,
        width=10cm,
        height=8cm,
        legend pos=north east,
    ]
    \addplot[
        blue,
        thick,
    ] table [x=time, y=u1, col sep=comma] {csv/control_inputs.csv};
    \addlegendentry{$u_1$}
    
    \addplot[
        red,
        thick,
    ] table [x=time, y=u2, col sep=comma] {csv/control_inputs.csv};
    \addlegendentry{$u_2$}
    \end{axis}
\end{tikzpicture}



    \caption{Variation of input versus time for single integrator system}
    \label{fig:inputs_time_single_integrator}
\end{figure}
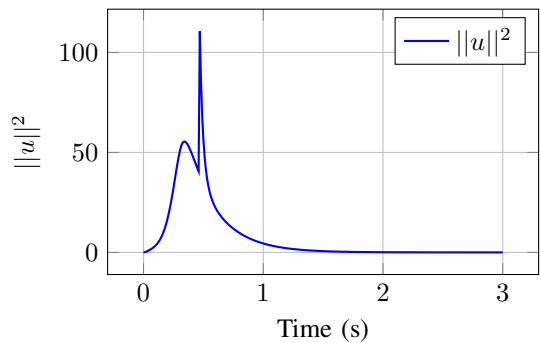
\begin{figure}
    \centering
\begin{tikzpicture}
    \begin{axis}[
    scale=0.55,
        xlabel={Time (s)},
        ylabel={$||u||^2$},
        title={},
        grid=both,
        legend pos=north east,
        width=12cm,
        height=8cm,
    ]

    \addplot[
        blue,
        thick,
    ] table [x=time, y=control_input_norm, col sep=comma] {csv/control_input_norm.csv};
    \addlegendentry{$||u||^2$}

    \end{axis}
\end{tikzpicture}
\caption{Variation of input norm versus time}
    \label{fig:input_norm_single_integrator}
\end{figure}
\begin{figure}
    \centering
\begin{tikzpicture}
    \begin{axis}[
    scale=0.55,
        xlabel={Time (s)},
        ylabel={$h(\bx,\bu)$},
        title={},
        legend pos=north east,
        grid=both,
        width=12cm,
        height=8cm,
    ]

    \addplot[
        blue,
        thick,
    ] table [x=time, y=h1, col sep=comma] {csv/ubfs.csv};



    \end{axis}
\end{tikzpicture}
\caption{Variation of $h(\bx,\bu)$ versus time for the single integrator system}
    \label{fig:ubf_single_integrator}
\end{figure}
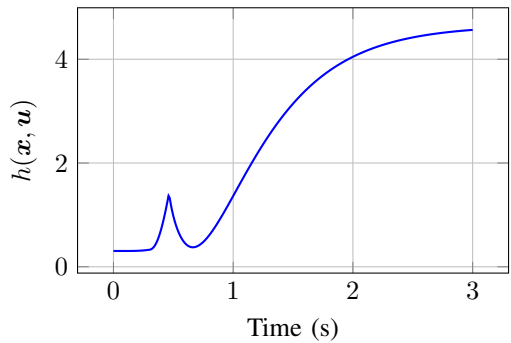

We first consider a robot with single integrator dynamics given by $\dot{\bx} = \bu$, where $\bx = [x_1, x_2]^\mathrm{T} \in \mathbb{R}^2$ represents the position of the robot in the plane, and $\bu = [u_1, u_2]^\mathrm{T} \in \mathbb{R}^2$ is the control input corresponding to its velocity vector.
The control objective is to steer the robot from $\bx_0 = [0.5, 1]^\mathrm{T}$ to $\bx_{\text{goal}} = [4.5, 4.5]^\mathrm{T}$ while avoiding collisions with obstacles and satisfying input constraints. The integral controller is computed via the Newton-Raphson flow (Section \ref{subsec:integral_controller_newton_raphson_flow}) where $T=0.55s$ in \eqref{eqn:y_tilde} and $\alpha=25$ in \eqref{eqn:int_controller_raphson_flow}. However, this integral controller does not inherently guarantee safety or input constraint satisfaction.

To enforce safety, we define three barrier functions corresponding to the obstacles and input constraints given by $h_1(\bx) = (x_1 - 3)^2 + (x_2 - 3)^2 - 0.4$, $h_2(\bx) = (x_1 - 1.5)^2 + (x_2 - 1.5)^2 - 0.25$, and $h_3(\bu) = 120 - u_1^2 - u_2^2$ (for input constraints).
For the UBF, we choose $\beta = 10$, $m_1=m_2=m_3=1$ and $m=2$ (see Section \ref{subsec:high_order_ubf}).
At each time step, we solve an UBF-QP to compute an auxiliary control input $\bv$ that modifies the nominal integral control to ensure safety and input constraint satisfaction. 
We simulate the system over a time horizon of $3$ seconds with a time step of $\Delta t = 0.01s$, resulting in $N = 300$ simulation steps. The initial control input is set to zero. At each time step, we compute the nominal control, evaluate the barrier functions, and solve the QP to find the optimal $\bv$. The control input and state are updated using forward Euler integration.

Figures \ref{fig:traj_single_integrator} and \ref{fig:input_norm_single_integrator} shows that the robot reaches the goal position while avoiding both obstacles and respecting the input constraints respectively. 



\subsection{Double Integrator System with HO-UBF}

We consider the double integrator dynamics given by
\begin{align}
\dot{x}_1=x_3,\;\;\dot{x}_2=x_4,\;\;\dot{x}_3=u_1,\;\dot{x}_4=u_2
\label{eqn:double_integrator_dynamics}
\end{align}
with $\bx=[x_1,\;x_2,\;x_3,\;x_4]^\mathrm{T}$ representing the position and velocity states in 2D, and $\bu=[u_1,\;u_2]^\mathrm{T}$ being the control input corresponding to accelerations in the $x_1$ and $x_2$ directions.
The control objective is similar to the previous example, i.e., to navigate the robot from the initial position $\bx_0 = [0.5, 1, 0, 0]^\mathrm{T}$ (starting from rest) to the goal position $\bx_{\text{goal}} = [4.5, 4.5]^\mathrm{T}$ while avoiding obstacles and respecting input constraints. For the integral controller, we choose $T=0.35s$ (in \eqref{eqn:y_tilde}) and $\alpha=35$ (in \eqref{eqn:int_controller_raphson_flow}). 

The barrier functions $h_1(\bx)$, $h_2(\bx)$, and $h_3(\bu)$ are used. For UBF, we choose $\beta = 20$.
Note that, for \eqref{eqn:double_integrator_dynamics}, the two individual barrier functions are of relative degree two with respect to the control input. To effectively enforce safety, we implement the HO-UBF with an order of $m_1 = 2$, $m_2=2$, $m_3=1$ and $m=2$ (see Section \ref{subsec:high_order_ubf}).
At each time step, we solve the UBF-QP to compute the auxiliary control input $\bv$ to satisfy the safety constraints. 
We simulate the system over a time horizon of $30s$ with a time step of $\Delta t = 0.001s$, resulting in $N = 30,000$ simulation steps.

Fig. \ref{fig:traj_double_integrator} shows the trajectory of the double integrator system \eqref{eqn:double_integrator_dynamics} in the presence of two obstacles. In addition, as shown in Fig. \ref{fig:input_norm_double_integrator}, the input constraints are respected.
\begin{figure}
    \centering
\begin{tikzpicture}
    \begin{axis}[
    scale=0.6,
        xlabel={$x_1$},
        ylabel={$x_2$},
        title={},
        legend pos= south east,
        grid=both,
        axis equal image,
        width=12cm,
        height=12cm,
    ]

    \addplot[
        blue,
        thick,
    ] table [x=x1, y=x2, col sep=comma] {csv/trajectory_double_int.csv};
    \addlegendentry{Trajectory}

    \addplot[
        only marks,
        mark=*, 
        mark size=3pt,
        red,
    ] coordinates {(5,4)};
    \addlegendentry{Goal Position}

    \addplot[
        red,
        fill=red, 
        fill opacity=0.3,
        draw=black,
        domain=0:360,
        samples=100,
        parametric,
    ] 
    (
        {3 + 0.632455532*cos(x)}, 
        {3 + 0.632455532*sin(x)}
    );
    \addlegendentry{Obstacle 1}

    \addplot[
        green,
        fill=green, 
        fill opacity=0.3,
        draw=black,
        domain=0:360,
        samples=100,
        parametric,
    ] 
    (
        {1.5 + 0.5*cos(x)}, 
        {1.5 + 0.5*sin(x)}
    );
    \addlegendentry{Obstacle 2}

    \end{axis}
\end{tikzpicture}    \caption{Trajectory of double integrator system}
    \label{fig:traj_double_integrator}
\end{figure}
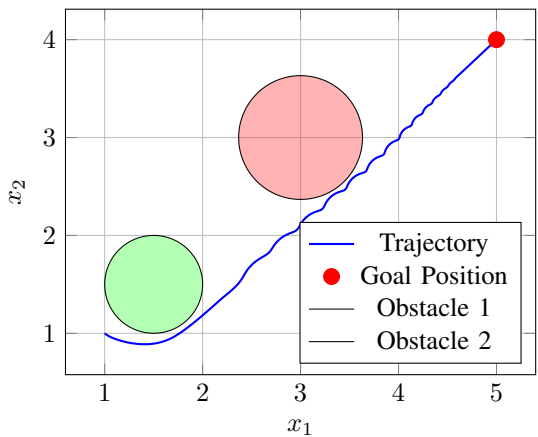

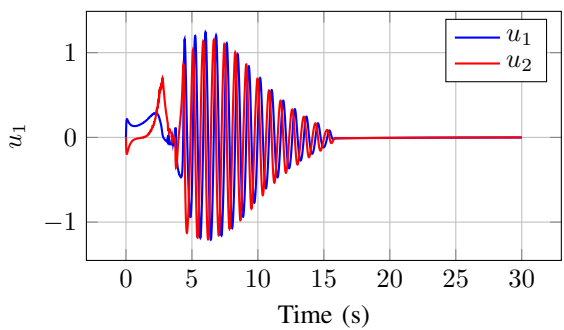
\begin{figure}
    \centering
\begin{tikzpicture}
    \begin{axis}[
    scale=0.75,
        xlabel={Time (s)},
        ylabel={$u_1$},
        title={},
        grid=both,
        width=10cm,
        height=6cm,
        legend pos=north east,
    ]
    \addplot[
        blue,
        thick,
    ] table [x=time, y=u1, col sep=comma] {csv/control_inputs_double_int.csv};
    \addlegendentry{$u_1$}
    \addplot[
        red,
        thick,
    ] table [x=time, y=u2, col sep=comma] {csv/control_inputs_double_int.csv};
    \addlegendentry{$u_2$}
    \end{axis}


\end{tikzpicture}    \caption{Variation of control input versus time for the double integrator system}
    \label{fig:inputs_time_double_integrator}
\end{figure}

\begin{figure}
    \centering
\begin{tikzpicture}
    \begin{axis}[
    scale=0.55,
        xlabel={Time (s)},
        ylabel={$||u||^2$},
        title={},
        grid=both,
        legend pos=north east,
        width=12cm,
        height=8cm,
    ]

    \addplot[
        blue,
        thick,
    ] table [x=time, y=control_input_norm, col sep=comma] {csv/control_input_norm_double_int.csv};
    \addlegendentry{$||u||^2$}

    \end{axis}
\end{tikzpicture}
\caption{Norm of control input versus time}
    \label{fig:input_norm_double_integrator}
\end{figure}
\subsection{Quadrotor System}
Finally, we consider the quadrotor system where the dynamics is given by
\begin{align}
\begin{bmatrix}
\dot{\bx}_p \\
\dot{\bx}_\theta \\
\dot{\bx}_v \\
\dot{\bx}_\omega
\end{bmatrix} &=
\begin{bmatrix}
\bx_v \\
\bx_\omega \\
\frac{u_1}{m_q}R(\bx_\theta)e_3 - ge_3 \\
I^{-1}(\bu_{2:4} - \bx_\omega \times I\bx_\omega)
\end{bmatrix}
\label{eqn:quadrotor_dynamics}
\end{align}
\begin{figure}
    \centering
\begin{tikzpicture}
    \begin{axis}[
    scale=0.5,
        xlabel={Time (s)},
        ylabel={$||\bu||^2$},
        title={},
        grid=both,
        legend pos=north east,
        width=15cm,
        height=8cm,
    ]

    \addplot[
        blue,
        thick,
    ] table [x=time, y=control_input_norm, col sep=comma] {csv/control_input_norm_quad.csv};
    \addlegendentry{$||u||$}

    \end{axis}
\end{tikzpicture}
\caption{Norm of control input versus time}
    \label{fig:input_norm_quad}
\end{figure}
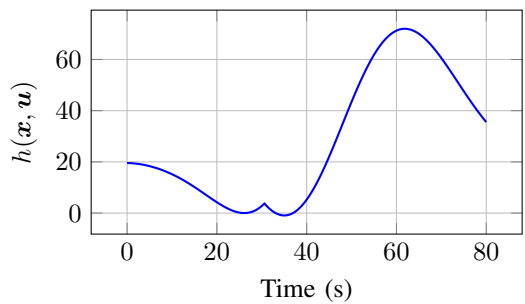
\begin{figure}
    \centering
\begin{tikzpicture}
    \begin{axis}[
    scale=0.55,
        xlabel={Time (s)},
        ylabel={$h(\bx,\bu)$},
        title={},
        legend pos=north east,
        grid=both,
        width=12cm,
        height=7cm,
    ]

    \addplot[
        blue,
        thick,
    ] table [x=time, y=h1, col sep=comma] {csv/ubfs_quad.csv};



    \end{axis}
\end{tikzpicture}
\caption{Variation of $h(\bx,\bu)$ versus time for the quadrotor system}
    \label{fig:ubf_quadrotor}
\end{figure}
where $\bx_p=[x,\;y,\;z]^\mathrm{T}$ and $\bx_v=[v_x,\;v_y,\;v_z]^\mathrm{T}$ represent the quadrotor's position and velocity, respectively. The orientation and angular velocity are denoted by $\bx_\theta=[\phi,\;\theta,\;\psi]^\mathrm{T}$ and $\bx_\omega=[\omega_x,\;\omega_y,\;\omega_z]^\mathrm{T}$. The mass of the quadrotor is represented by $m_q$, while $R(\bx_\theta)$ is the rotation matrix that transforms body-fixed coordinates to inertial coordinates. The vector $e_3 = [0, 0, 1]^\mathrm{T}$ represents the unit vector, $g$ denotes the gravitational acceleration and $I$ is the inertia matrix of the quadrotor. The control input vector $\bu = [u_1, u_2, u_3, u_4]^\mathrm{T}$ consists of four components, where $u_1$ represents the total thrust, and $\bu_{2:4} = [u_2, u_3, u_4]^\mathrm{T}$ represents the moments applied to the quadrotor.
The simulation was run for $80s$ with a time step of $0.005s$.
To construct the UBF, we choose the following barrier functions $h_1(\bx) = (x_1 - 3)^2 + (x_2 - 3)^2 +(x_3-3)^2 - 0.4$, $h_2(\bx) = (x_1 - 1.5)^2 + (x_2 - 1.5)^2 +(x_3-3)^2 - 0.25$, and $h_3(\bu) = 200 - u_1^2 - u_2^2$ (for input constraints).
The initial state was set to $\bx_0 = [0, 0, 0.5, 0, 0, 0, 0, 0, 0, 0, 0, 0]^\mathrm{T}$ and the goal state to be $\bx_g = [5, 5, 5, 0, 0, 0, 0, 0, 0, 0, 0, 0]^\mathrm{T}$.
We used an integral controller gain of $\alpha = 25.0$ and set the UBF orders to be $m_1=2$, $m_2=2$, $m_3=1$ and $m=1$ (see Section \ref{subsec:high_order_ubf}). The $K_\infty$ function for UBF is set to $\alpha(h)=3h$.
\begin{figure}
    \centering
\begin{tikzpicture}
    \begin{axis}[
    scale=0.5,
        view={60}{30}, 
        xlabel={$x$},
        ylabel={$y$},
        zlabel={$z$},
        title={},
        legend pos=south east,
        grid=both,
        width=15cm,
        height=12cm,
        xmin=-1, xmax=7,
        ymin=-1, ymax=7,
        zmin=0, zmax=6,
    ]

    \addplot3[
        blue,
        thick,
        mark=none,
        filter discard warning=false,
        filter point/.code={
            \ifnum\coordindex>0
                \pgfmathtruncatemacro{\modulus}{int(mod(\coordindex,100))}
                \ifnum\modulus=0
                \else
                    \def\pgfplotspointmeta{discard}
                \fi
            \fi
        }
    ] table [
        x=x, 
        y=y, 
        z=z, 
        col sep=comma
    ] {csv/trajectory_quad.csv};
    \addlegendentry{Trajectory}

    \addplot3[
        only marks,
        mark=*,
        mark size=3pt,
        red,
    ] coordinates {(5,5,5)};
    \addlegendentry{Goal Position}

    \addplot3[
        red,
        fill=red,
        fill opacity=0.3,
        draw=black,
        samples=20,
        samples y=10,
        domain=0:360,
        y domain=0:180,
    ]
    ({3 + 0.632455532*sin(y)*cos(x)},
     {3 + 0.632455532*sin(y)*sin(x)},
     {3 + 0.632455532*cos(y)});
    \addlegendentry{Obstacle 1}

    \addplot3[
        green,
        fill=green,
        fill opacity=0.3,
        draw=black,
        samples=20,
        samples y=10,
        domain=0:360,
        y domain=0:180,
    ]
    ({1.5 + 0.5*sin(y)*cos(x)},
     {1.5 + 0.5*sin(y)*sin(x)},
     {2.0 + 0.5*cos(y)});
    \addlegendentry{Obstacle 2}


    \end{axis}
\end{tikzpicture}    \caption{Trajectory of quadrotor system}
    \label{fig:traj_quad}
\end{figure}
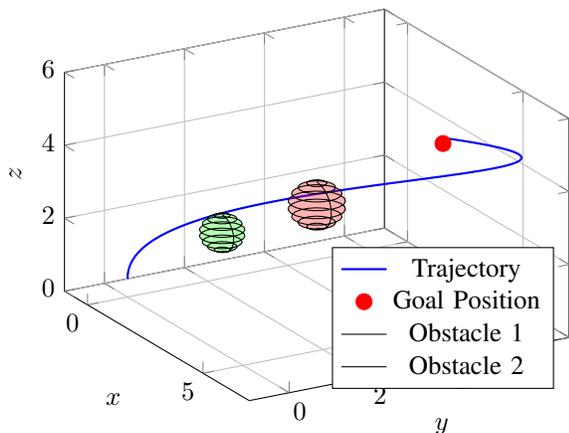


Fig. \ref{fig:traj_quad} shows the 3D trajectory of the quadrotor, demonstrating the UBF-QP controller's ability to navigate in 3D space while avoiding spherical obstacles. As seen from Figs. \ref{fig:input_norm_quad} , the input constraints are respected only if the barrier function for inputs is used while constructing a UBF.

\section{Conclusion\label{sec:conclusion}}
In this paper, we proposed the Universal Barrier Function (UBF), a single scalar-valued, continuously differentiable function designed based on which one can design controllers that can account for both safety and stability for controlled nonlinear systems subject to input constraints. Next, we proposed UBF-based quadratic programs (UBF-QP) to synthesize safe and stabilizing control inputs while satisfying complex state and input constraints. This approach is further extended to systems with higher relative degrees. Future work will include addressing the challenges of deploying the UBF-QP based controllers in the real world. 

\bibliography{main.bib}

\end{document}